\documentclass[reqno]{amsart}

\usepackage{amsmath}
\usepackage{amsfonts}
\usepackage{amssymb}
\usepackage{mathrsfs}
\usepackage[colorlinks,citecolor=blue,linkcolor=blue]{hyperref}
\usepackage[sort&compress,numbers]{natbib}

\newtheorem{theorem}{Theorem}
\theoremstyle{plain}

\newtheorem{definition}{Definition}

\newtheorem{lemma}{Lemma}

\newtheorem{proposition}{Proposition}
\newtheorem{remark}{Remark}

\numberwithin{equation}{section}

\numberwithin{equation}{section}
\numberwithin{theorem}{section}
\numberwithin{definition}{section}
\numberwithin{lemma}{section}
\numberwithin{proposition}{section}
\numberwithin{corollary}{section}

\newcommand{\naturals}{\mathbb{N}}
\newcommand{\reals}{\mathbb{R}}
\newcommand{\complexes}{\mathbb{C}}
\newcommand{\integers}{\mathbb{Z}}

\newcommand{\domain}[1]{\mathrm{Dom} \, #1}

\newcommand{\spec}[1]{\mathrm{\sigma}(#1)}

\newcommand{\tnorm}[1]{\| #1 \|_{1}}
\newcommand{\fnorm}[1]{\| #1 \|_{2}}

\newcommand{\tr}[1]{\mathrm{tr} \, #1 }
\newcommand{\comm}[2]{[#1,#2]}
\newcommand{\acomm}[2]{\{#1,#2\}}
\newcommand{\ad}[1]{\comm{#1}{\,\cdot \,}}
\newcommand{\adj}{\star}
\newcommand{\hadj}{*}
\newcommand{\dual}{\prime}
\newcommand{\tp}[1]{\mathrm{TP}( #1 )}
\newcommand{\cp}[1]{\mathrm{CP}( #1 )}
\newcommand{\cptp}[1]{\mathrm{CP,\,TP}( #1 )}

\newcommand{\matr}[1]{\complexes^{#1 \times #1}}
\newcommand{\matrd}{\matr{d}}

\newcommand{\trs}[1]{\mathrm{tr}_{\bar{\otimes}} \, #1 }
\newcommand{\trl}[1]{\mathrm{tr}_{\mathscr{L}^{2}} \, #1 }

\renewcommand{\vec}[1]{\mathbf{#1}}

\DeclareMathOperator*{\esssup}{ess\,sup}

\begin{document}

\title[On Howland's time-independent formulation of CP-divisible ...]{On Howland's time-independent formulation of CP-divisible quantum evolutions}

\author{Krzysztof Szczygielski}
\address[K. Szczygielski]{Institute of Theoretical Physics and Astrophysics, Faculty of Mathematics, Physics and Informatics, University of Gda\'{n}sk, Wita Stwosza 57, 80-308 Gda\'{n}sk, Poland}
\email{krzysztof.szczygielski@ug.edu.pl}

\author{Robert Alicki}
\address[R. Alicki]{International Centre for Theory of Quantum Technologies (ICTQT), University of Gda\'{n}sk, Wita Stwosza 63, 80-308 Gda\'{n}sk, Poland}
\email{robert.alicki@ug.edu.pl}

\begin{abstract}
We extend Howland time-independent formalism to the case of completely positive and trace preserving dynamics of finite dimensional open quantum systems governed by periodic, time dependent Lindbladian in Weak Coupling Limit, expanding our result from previous papers. We propose the Bochner space of periodic, square integrable matrix valued functions, as well as its tensor product representation, as the generalized space of states within the time-independent formalism. We examine some densely defined operators on this space, together with their Fourier-like expansions and address some problems related to their convergence by employing general results on Banach-space valued Fourier series, such as the generalized Carleson-Hunt theorem. We formulate Markovian dynamics in the generalized space of states by constructing appropriate time-independent Lindbladian in standard (Lindblad -- Gorini -- Kossakowski -- Sudarshan) form, as well as one-parameter semigroup of bounded evolution maps. We show their similarity with Markovian generators and dynamical maps defined on matrix space, i.e.~the generator still possesses a standard form (extended by closed perturbation) and the resulting semigroup is also completely positive, trace preserving and a contraction.
\end{abstract}


\maketitle

\section{Introduction}

Completely positive (CP) and trace preserving (TP) dynamics of open quantum systems governed by time-dependent generators recently began gaining an increasing attention worldwide. In particular, the concept of evolution of open system in the regime of Markovian approximation and under external perturbation of \emph{periodic} nature led to the formulation of both Markovian Master Equation (MME) and appropriate quantum dynamical map in \cite{Alicki2006b}. In this paper we present an extension of our previous results \cite{Alicki2006b,Szczygielski2014} on CP-divisible dynamics governed by periodic generator in standard form by introducing the so-called \emph{time-independent formalism}. Such approach was originally invented and applied in the case of unitary, reversible dynamics under periodic self-adjoint bounded Hamiltonian by Shirley \cite{Shirley_1965}, Sambe \cite{Sambe_1973} and Howland \cite{Howland1974,Howland1979} and was successfully utilized in e.g. description of nuclear magnetic resonance (NMR) \cite{Ernst2005,Leskes2010,Scholz2010}, theory of decoherence suppresion \cite{Bach2014}, general laser spectroscopy \cite{Ho1983,Chu2004} and others. Time-independent formalism employs the Floquet theory in order to lift the \emph{dynamic} description of ordinary differential equations (ODEs), such as Schroedinger equation, to in a sense \emph{static}, algebraic one constructed in the infinite-dimensional \emph{generalized space of states}. The original ODE is then mapped to an eigenequation of some unbounded self-adjoint linear operator, the \emph{Floquet Hamiltonian}, acting on this large space. Moreover, one can show that the semigroup generated by this operator may be actually utilized to construct the solution of original ODE. Here we present a construction of the generalized space of states and appropriate dynamics in case of open quantum system described by Markovian Master Equation.

This paper is structured as follows: in Section \ref{sect:OpenQuantumSystems} we provide a concise recollection of most important results regarding CP-divisible dynamics of $d$-dimensional open quantum systems with periodically modulated Hamiltonian, preceded by short introduction to Floquet theory. The main result of this paper is presented in Section \ref{sect:SambeFloquet}, where we propose a formal construction of the time-independent formalism for CP-divisible open quantum systems. Starting with analysis of Banach space-valued Fourier series, we construct the generalized space of states as a certain Bochner space (definition \ref{def:GeneralizedSpace}), as well as appropriate representations of algebra of bounded operators. Then, we derive generalized CP-divisible quantum dynamical semigroup in Section \ref{subsect:MEinSspace} (theorem \ref{thm:WtauCPTPcontraction}) and demonstrate its relation to CP-divisible dynamical maps on matrix space. Finally, we elaborate (section \ref{sec:FourierExpansions}) on Fourier-like expansions of operators acting on the generalized space of states (which is a common and popular representation in realm of traditional time-independent formalism) and address their convergence.

Throughout this paper, $\matrd$ will denote a linear space of complex square matrices of size $d$ and we will use boldface to denote vectors in $\complexes^d$ or functions with values in $\complexes^d$ or $\matrd$. $\|\cdot\|_{\infty}$ will generally denote induced operator norm (of matrix or general linear operator), while $\| \cdot \|_{L^p}$ or $\| \cdot \|_{\mathscr{L}^p}$, $p\in[1,\infty ]$, will be reserved for functions (either complex or linear space valued) and denote appropriate $L^p$-norm, with $p=\infty$ reserved for supremum norm. Conjugate transpose of matrix $a\in\matrd$ will be denoted as $a^\hadj$ and we will use a notation $x^\adj$ to denote involution in certain spaces; for $T$ being a linear map on Banach space $\mathcal{X}$, $T^\dual$ will denote a dual map and $\mathcal{X}^\dual$ will be a topological dual space. Time derivative will be interchangeably denoted as $\frac{dy}{dt}$ or $\dot{y}$. We will emphasize any other notational conventions, as needed.

\section{Periodically controlled open quantum systems}
\label{sect:OpenQuantumSystems}

\subsection{Floquet theory}
\label{subsect:FloquetTheory}

Let $\mathbb{T} \simeq \reals / T \integers$, $T > 0$, denote a circle group and let a matrix-valued function $t \mapsto A(t) \in \matrd$, $t\in\mathbb{T}$, be continuous. Consider the initial value problem of a form
\begin{equation}\label{eq:CauchyProblem}
	\dot{x}(t) = A(t) x(t), \quad x(0) = x_0,
\end{equation}
where $x(t), x_0 \in \complexes^d$. Then, there exists a function $t\mapsto\Phi(t)$ such that each solution to \eqref{eq:CauchyProblem} is of a form $x(t) = \Phi(t) c$ for some vector $c\in\complexes^d$. $\Phi(t)$ is called the \emph{fundamental matrix solution} of the ODE and by construction, $\dot{\Phi}(t) = A(t) \Phi(t)$ and $\Phi(t)$ is invertible for each $t$. Since fundamental solution is always non-unique, one may additionally require $\Phi(0) = I$; in this case, it is called \emph{principal}. Then, the following Floquet's theorem applies:

\begin{theorem}[\textbf{Floquet}]\label{thm:FloquetTheorem}
For a Cauchy problem of a form \eqref{eq:CauchyProblem} with continuous and $T$--periodic $A(t)$, there exist a $T$--periodic function $t\mapsto P(t) \in \matrd$ and constant $B\in\matrd$ such that the fundamental matrix solution is 
\begin{equation}\label{eq:FloquetNormalForm}
	\Phi(t) = P(t) e^{tB}.
\end{equation}
\end{theorem}
For proof, see e.g. \cite{Chicone2006}. Expression \eqref{eq:FloquetNormalForm} is known as \emph{Floquet normal form} of fundamental matrix solution and $e^{TB}$ is called \emph{monodromy matrix}. Assume $B$ to be diagonalizable, i.e.~that there exists a linearly independent system $(\varphi_i)_{i=1}^{d}$ in $\complexes^d$ satisfying eigenequations $B\varphi_i = \xi_i \varphi_i$, where $\{\xi_i\}_{i=1}^{d} = \spec{B}$ is the spectrum of $B$. Then, $\spec{e^{TB}} = \{e^{\xi_i T}\}_{i=1}^{d}\subset \complexes\setminus \{0\}$ and monodromy matrix is diagonalizable as well. Eigenvalues $\xi_i$ of $B$ are called the \emph{Floquet exponents} (real part of each $\xi_i$ is called \emph{Lyapunov exponent}), while eigenvalues $e^{\xi_i T}$ of monodromy matrix are also known as \emph{characteristic multipliers} of the ODE; much about asymptotic stability of solutions can be implied from the exact analysis of characteristic multipliers \cite{Chicone2006}. One can immediately spot however, that spectrum of $B$ is non-uniquely determined by form of the monodromy matrix. Indeed, for every $\xi_j \in \spec{B}$ there exists a countable set $\{\xi_{j,n} \in \spec{B_k}\}$,
\begin{equation}
	\xi_{j,n} = \xi_{j} + in\Omega , \quad \xi_j \in \spec{e^{TB}}, \, n\in\integers,
\end{equation}
for $\Omega = 2\pi/ T$, such that $e^{\xi_{j,n}T} = e^{\xi_j T}$, so in general there exist infinitely many Floquet exponents corresponding to the same set of characteristic multipliers. We will call $\bigcup_{j=1}^{d} \{\xi_{j,n} \}$ a \emph{set of shifted Floquet exponents}.

As $(\varphi_j)_{j=1}^{d}$ is a basis in $\complexes^d$, one can put $c = \sum_{j=1}^{d} c_j \varphi_j$ and
\begin{equation}\label{eq:ODEgenSolDecomp}
	x(t) = \Phi(t) c = \sum_{j=1}^{d} c_j \varphi_j (t),
\end{equation}
where $\varphi_j (t) = e^{t\xi_j}\phi_j (t)$, and $\phi_j(t) = P(t) \varphi_j$ are periodic functions, called the \emph{Floquet solutions} or \emph{Floquet states}. It is also not hard to prove that $(\varphi_j (t))_{j=1}^{d}$ and $(\phi_j (t))_{j=1}^{d}$ are \emph{bases} in $\complexes^d$ for every $t\in\reals$; to prove this, it suffices to show linear independence of both sets which comes easily from invertibility of $\Phi(t)$ and linear independence of $(\varphi_j )_{j=1}^{d}$.

\subsection{Periodically driven open quantum systems}

Consider the open quantum system described by Hilbert space $\complexes^d$, $d \geqslant 1$, equipped with standard Euclidean (dot) inner product $\vec{x} \cdot \vec{y} = \sum_{j=1}^{d} \overline{x_j} y_j$, $\vec{x},\vec{y} \in \complexes^{d}$. Let $B(\complexes^d) \simeq \matrd$ be the C*-algebra of all linear operators on $\complexes^d$, endowed with operator norm $\| \cdot \|_{\infty}$, induced by Euclidean norm, i.e.
\begin{equation}\label{eq:InducedOpNorm}
	\| a \|_\infty = \sup_{\|\vec{w}\|\leqslant 1} \| a\vec{w}\|, \quad \| \vec{w} \| = \sqrt{\vec{w}\cdot\vec{w}}, \quad \vec{w} \in \complexes^d ,
\end{equation}
and conjugate transpose $a \mapsto a^\hadj$ as involution. We endow $\matrd$ with the \emph{trace norm}
\begin{equation}
	\tnorm{a} = \tr{\sqrt{a^\hadj a}},
\end{equation}
which makes $(\matrd,\tnorm{\cdot})$ a Banach space and a Banach algebra, since it can be shown to satisfy inequalities
\begin{equation}\label{eq:InequalityOfMatrixNorms}
	\| a \|_{\infty} \leqslant \tnorm{a}, \qquad \tnorm{ab} \leqslant \| a \|_{\infty} \tnorm{b} \leqslant \tnorm{a}\tnorm{b}
\end{equation}
for any $a,b\in\matrd$. 

As usual, we describe the state of open system with a time-dependent \emph{density operator} $\rho_t \in \matrd$, i.e.~a Hermitian, positive-semidefinite matrix of trace one. The time evolution of $\rho_t$ will be given by sufficiently smooth mapping $t\mapsto \Lambda_t \in B(\matrd)$ such that
\begin{equation}
	\rho_t = \Lambda_t (\rho_0), \qquad \rho_t \in (\matrd)^{+}, \, \tnorm{\rho_t} = 1,
\end{equation}
where $\rho_0$ is some initial density matrix, $(\matrd)^{+}$ is a positive cone in $\matrd$ and linear map $\Lambda_t$ is \emph{completely positive and trace preserving} (CP, TP) on $\matrd$ for each $t\in [0,\infty)$, called the \emph{quantum dynamical map}. The associated \emph{propagator} $(t,s)\mapsto V_{t,s}$ will be required to satisfy the divisibility condition
\begin{equation}
	\Lambda_t = V_{t,s} \Lambda_s , \qquad 0 \leqslant s \leqslant t.
\end{equation}
Provided $\Lambda_{s}^{-1}$ exists, one easily obtains $V_{t,s} = \Lambda_{t} \Lambda_{s}^{-1}$. Then, $\Lambda_t$ is called \emph{CP-divisible} or \emph{Markovian} if and only if $V_{t,s}$ is CP, TP for any $s \in [0,t]$ for a given $t\in [0,\infty )$. We will focus solely on the case of \emph{differentiable} dynamical maps, i.e.~we invoke the common assumption of existence of time-dependent \emph{Lindbladian} $L_t$, such that $\Lambda_t$ satisfies the time-local MME, here presented in two equivalent forms, 
\begin{equation}\label{eq:MMELambda}
	\dot{\Lambda}_t = L_t \Lambda_t, \qquad \dot{\rho}_t = L_t (\rho_t),
\end{equation}
with initial conditions often stated as $\rho_{0}$ and $\Lambda_0 = I$. One can show \cite{Rivas2012,Chruscinski2014a} that if the MME is satisfied, then the resulting $\Lambda_t$ is CP-divisible if and only if $L_t$ admits the so-called \emph{standard (Kossakowski -- Lindblad -- Gorini -- Sudarshan)} form (we put $\hbar = 1$)
\begin{equation}\label{eq:StandardForm}
	L_{t}(\rho) = -i \comm{H_t}{\rho} + \sum_{k=1}^{d^2 - 1} \Big(V_{k,t} \rho V_{k,t}^{\hadj} - \frac{1}{2}\acomm{V_{k,t}^{\hadj}V_{k,t}}{\rho}\Big),
\end{equation}
where $H_t = H_{t}^{\hadj} \in \matrd$ is the effective Hamiltonian of the system and $( V_{k,t})$ is a sequence of time-dependent noise operators. No universal method of solving \eqref{eq:MMELambda} is known apart from general techniques involving Dyson or Magnus expansions \cite{Alicki2006a,Rivas2012,Blanes1998}. Nevertheless, if the function $t\mapsto L_t$ is \emph{constant} or \emph{periodic}, the exact form of solutions is known. In the former case, namely if $L_t = L$, the solution is the celebrated \emph{quantum dynamical semigroup} $\{e^{tL} : t \geqslant 0 \}$, i.e.~one-parameter, strongly continuous contraction semigroup of CP, TP maps on $\matrd$; then, $L$ is an infinitesimal generator of this semigroup given in standard form \eqref{eq:StandardForm}, however without any time dependence (see \cite{Lindblad1976,Gorini1976,Alicki2006a,Rivas2012,Davies1974,Davies1976} and references therein for details). The latter case, involving periodic time-dependence of $L_t$ is more involved, as we elaborate in next section.

\subsection{Examples of periodic Lindbladians}
\label{sect:ExamplesofperiodicLindbladians}

At least two classes of \emph{periodic} Lindbladians, which were derived from the underlying Hamiltonian dynamics of the open system weakly interacting with a quantum stationary reservoir, are known in the literature. The presented examples are used in quantum thermodynamics literature to describe various types of heat engines or refrigerators.

The first one was presented in e.g. \cite{Alicki2006b,Szczygielski2013,Szczygielski2014}; here the open quantum system was described by a time-dependent, periodic Hamiltonian $H_t = H_{t+T}$ and coupled to a large reservoir described in the thermodynamic limit. The interaction Hamiltonian was given by usual expression
\begin{equation}
	H_{\mathrm{i}} = \sum_{k\in\mathcal{I}} S_k \otimes R_k ,
\end{equation}
where $S_k \in \matrd$, and $R_k$ are the reservoir's observables. Beside the standard assumptions allowing for application of Weak Coupling Limit (WCL) techniques, one assumes that the modulation is \emph{fast}, i.e.~its frequency $\Omega$ is comparable to the relevant Bohr frequencies of the system Hamiltonian. In this case one first applies the Floquet decomposition of the system unitary dynamics defined by propagator $u_t$,
\begin{equation}\label{eq:TheSchroedingerEquation}
	\dot{u}_t = -i H_t u_t, \quad H_{t+T} = H_t, \, u_0 = I.
\end{equation}
Due to Floquet theorem (Theorem \ref{thm:FloquetTheorem}), a principal fundamental matrix solution $u_t$ is then
\begin{equation}\label{eq:UnitaryEvolutionOperator}
	u_t = p_t e^{-it\bar{H}},
\end{equation}
where $p_t$ is unitary and periodic, $p_0 = I$, and $\bar{H}$ is Hermitian and called the \emph{averaged Hamiltonian}. A set of \emph{Bohr quasifrequencies} is defined as
\begin{equation}
	\{\omega\} = \{\epsilon - \epsilon^\prime : \epsilon, \epsilon^\prime \in \spec{\bar{H}}\} .
\end{equation}
The Schroedinger picture system dynamics is denoted by $U_t (a) = u_t a u_t^{\hadj}$ and its periodic component by $P_{t}(a) = p_{t} a p_{t}^{\hadj}$. Then, applying WCL technique, one can derive the MME of the form
\begin{equation}\label{eq:LtFloquet}
	\dot{\rho}_t = L_t(\rho_t) = -i\comm{H_t}{\rho_t} + (P_t \, K \, P_{t}^{-1})(\rho_t)
\end{equation}
with solution in a factorized form
\begin{equation}\label{eq:LambdaGeneral}
	\Lambda_t = U_t e^{tK} = P_t e^{t\bar{L}}, \quad \bar{L} = -i\ad{\bar{H}} + K.
\end{equation}
The Davies type structure of the semigroup generator $K$ is 
\begin{equation}\label{eq:LindInterPicture}
	K(\rho) = \sum_{k,k' \in \mathcal{I}}\sum_{q\in\integers}\sum_{\{\omega\}}\gamma_{kk'}(\omega + q\Omega) \Big( S_{kq\omega} \rho S_{k'q\omega}^{\hadj} - \frac{1}{2}\acomm{S_{kq\omega}^{\hadj}S_{k'q\omega}}{\rho}\Big).
\end{equation}
Operators $S_{kq\omega}$ may be then shown to satisfy relations
\begin{equation}\label{eq:SkomegaRelations}
	\comm{\bar{H}}{S_{kq\omega}} = \omega S_{kq\omega}, \qquad S_{kq\omega}^{\hadj} = S_{k,-q,-\omega},
\end{equation}
and the positive definite matrix $[\gamma_{kk'}(\omega + q\Omega)]$ is, as usually, the Fourier transform of suitable reservoir correlation matrix taken at $\omega + q\Omega$. 
By direct computation one can show that the semigroup generator $K$ \eqref{eq:LindInterPicture} commutes with the generator $-i\ad{\bar{H}}$. This property leads to product forms of the dynamics \eqref{eq:LambdaGeneral}. A simple diagonalization procedure transforms \eqref{eq:LindInterPicture} into the standard form and then the generator in MME \eqref{eq:LtFloquet} takes the form \eqref{eq:StandardForm} with periodic $H_t$ and $V_k^t$.

The second class can be derived for \emph{slowly varying} system Hamiltonians using WCL approach \cite{Davies1978,Alicki1979,Alicki1979a}.  Namely, following Davies and Spohn \cite{Davies1978}, we assume that the full Hamiltonian of system plus reservoir is of a form
\begin{equation}
	H_{t}^{\text{sr}} = H_{\lambda^2 t} \otimes I + I \otimes H_{\text{r}} + \lambda H_{\text{int.}},
\end{equation}
where $H_{\lambda^2 t}$, $H_{\text{r}}$ and $H_{\text{int.}}$ are self-adjoint Hamiltonians of the system, reservoir and interaction, respectively, and $0 < \lambda \ll 1$ is a small dimensionless parameter (hence the weak coupling). Note that for technical reasons, the system's part of $H_{t}^{\text{sr}}$ is actually assumed to be a function of \emph{rescaled} time parameter $\lambda^2 t$, which then allows to obtain a rigorous and well-behaved \emph{adiabatic limit} of resulting dynamics via application of quantum adiabatic theorem (see \cite{Davies1978} for more rigorous treatment). This rescaling implies that the system's Hamiltonian changes \emph{slowly} compared to any other relevant time change within the model, i.e. by increasing time $t$ by $\Delta t$ small enough, the change of $H_{\lambda^2 t}$ is
\begin{equation}\label{eq:AdiabaticChangeOfH}
	H_{\lambda^2 (t+\Delta t)} - H_{\lambda^2 t} \simeq \lambda^2 \left. \frac{dH_s}{ds}\right|_{\lambda^2 t} \Delta t + \mathcal{O}(\lambda^4),
\end{equation}
which is small (in norm sense) due to factor $\lambda^2 \ll 1$. Then, applying techniques similar to those in \cite{Davies1978}, one can derive the MME of the following form
\begin{equation}\label{eq:AdFloquet}
	\dot{\rho}_t = -i\comm{H_t}{\rho_t} + \sum_{k,k'}\sum_{\{\omega_t\}}\gamma_{kk'}(\omega_t) \Big( S_{k\omega_t} \rho_t S_{k'\omega_t}^{\hadj} - \frac{1}{2}\acomm{S_{k\omega_t}^{\hadj}S_{k'\omega_t}}{\rho_t}\Big),
\end{equation}
where $\{\omega_t\}$ are time-dependent Bohr frequencies of $H_t$ and matrices $S_{k\omega_t}$ satisfy
\begin{equation}\label{eq:AdSkomegaRelations}
	\comm{H_t}{S_{k\omega_t}} = \omega_t S_{k\omega_t}, \qquad S_{k\omega_t}^{\hadj} = S_{k,-\omega_t}.
\end{equation}
Formula \eqref{eq:AdiabaticChangeOfH} then implies, that all Bohr frequencies again are slowly-varying, i.e.~$\Delta \omega_t \simeq \lambda^2 \dot{\omega}_{\lambda^2 t} \Delta t$. It follows from \eqref{eq:AdSkomegaRelations} that at any moment of time the Hamiltonian part commutes with the dissipative one. The WCL procedure assures that if $H_t$ is then additionally assumed \emph{periodic}, matrices $S_{k\omega_t}$ are periodic and the whole Lindbladian in \eqref{eq:AdFloquet} is also periodic.

\section{Howland time-independent formalism}
\label{sect:SambeFloquet}

In this section we will present a formal construction of time-independent formalism for open quantum system governed by MME \eqref{eq:LtFloquet}. The formalism will exhibit some similarities with usual quantum-dynamical approach, e.g. complete positivity, trace preservation and contractivity of evolution maps. We will largely adapt the approach by Sambe \cite{Sambe_1973} and Howland \cite{Howland1974,Howland1979} as mentioned earlier, who introduced the idea of enlarged Hilbert space of states and generalized dynamics in order to find solutions of Schroedinger equation with periodic Hamiltonian as an extension of method proposed by Shirley \cite{Shirley_1965}. We will not, however, elaborate much on the unitary (Hamiltonian) time independent formalism; interested readers should refer to literature, e.g. \cite{Shirley_1965,Sambe_1973,Howland1974,Howland1979,Scholz2010,Ernst2005,Ho1983,Chu2004} and references therein, for further details.

The current section contains our main results. In section \ref{sect:FourierAnalysis} we propose a formal candidate for generalized space of states (in definition \ref{def:GeneralizedSpace}) suitable for open systems theory. We propose this space to be a certain \emph{Bochner space} of periodic functions with values in Banach space $(\matrd , \tnorm{\cdot})$ of matrices furnished with trace norm. Section \ref{sec:IsomorphicRep} contains a formal construction of generalized space of states from its underlying \emph{tensor product} structure, i.e.~we show that in principle it may be constructed as a Banach space completion of certain algebraic tensor product with respect to Bochner space norm. We emphasize here, that both the definition of the space and its tensorial structure are stated in the similar spirit as in traditional approach, however they are of \emph{different} nature (see remark \ref{remark:Comparison} for comparisons). We then formulate notions of \emph{complete positivity} and \emph{trace preservation} of maps acting on generalized space (based on definitions \ref{eq:DefinitionNpos}, \ref{def:TraceOnGenSpace}), followed by miscellaneous results in section \ref{sec:ExamplesOfOperators}. Next, we formulate a definition of Floquet Lindbladian (being a counterpart of Floquet Hamiltonian from Howland approach) and show some properties of corresponding generated $C_0$-semigroup in theorem \ref{thm:WtauCPTPcontraction}, which may be considered as the main result of the paper. The last part of section \ref{sect:SambeFloquet} concerns our results for convergence of certain Fourier-like expansions of operators acting on generalized space of states.

For convenience we will require some regularity of Hamiltonian $H_t$, namely we will assume $t\mapsto H_t$ is periodic and of piecewise continuous first derivative. This assumption is quite mild and will be of particular importance for convergence of certain operator series later on.

\subsection{Generalized space of states}
\label{sect:FourierAnalysis}

Let $(V, \Sigma, \mu)$ be a finite measure space and let $(\mathcal{X}, \mathcal{X}^\dual)$ be a dual pair for Banach space $(\mathcal{X},\|\cdot \|)$, with duality pairing $\varphi (x) \equiv (\varphi, a)_{\mathcal{X},\mathcal{X}^\dual}$, $a\in \mathcal{X}$, $\varphi \in \mathcal{X}^\dual$. We denote by $\mathscr{L}^p (V,\mathcal{X})$, $p \in [1,\infty ]$, the \emph{Bochner space} of $\mu$-measurable, $p$-integrable, $\mathcal{X}$-valued functions on $V$, complete with respect to $\mathscr{L}^p$-norm
\begin{equation}
	\| f \|_{\mathscr{L}^p} = \Big( \int_{V} \| f \|^{p} \, d\mu \Big)^{1/p}, \qquad \| f \|_{\mathscr{L}^\infty} = \esssup_{x\in V}{\| f(x) \|} .
\end{equation}
Naturally, $\mathscr{L}^p (V, \mathcal{X})$ is a Banach space for every $p\in [1,\infty)$. One shows \cite{Diestel_1977} that if $\mathcal{X}$ has a \emph{Radon-Nikodym property} with respect to $(V,\Sigma,\mu)$, then $\mathscr{L}^p (V, \mathcal{X})^\dual$ is isometrically isomorphic to $\mathscr{L}^q (V, \mathcal{X}^\dual)$ for $p^{-1}+q^{-1}=1$  and the duality pairing $\mathscr{L}^p \times \mathscr{L}^q \mapsto (f,g)_{\mathscr{L}^p} \in \complexes$ takes a form of Lebesgue-Stieltjes integral
\begin{equation}
	(f,g)_{\mathscr{L}^p} = \int_{V} (f(x), g(x))_{\mathcal{X},\mathcal{X}^\dual} \, d\mu(x).
\end{equation}
Let $(V, \Sigma, \mu) = (\mathbb{T},\mathcal{B}(\mathbb{T}),T^{-1}\lambda)$ with $\Sigma = \mathcal{B}(\mathbb{T})$ the $\sigma$-algebra of Borel subsets of $\mathbb{T}$ and $\mu = T^{-1}\lambda$ a normalized Lebesgue measure, $T^{-1}\lambda(\mathbb{T})=1$. Also, let $(e_n)_{n\in\integers}$, $e_n (t) = e^{in\Omega t}$ denote the Fourier basis in space $L^p (\mathbb{T})$, $p \in (1,\infty )$ and let $D_n = \sum_{|k|\leqslant n} e_k$ be the Dirichlet kernel. The \emph{$n$-th partial Fourier sum operator} $D_n \, *$ is given by convolving $D_n$ with $f$ as
\begin{equation}
	f \mapsto D_n * f = \sum_{|k|\leqslant n} f_k \odot e_{k}, \qquad f_k = \frac{1}{T} \int_{\mathbb{T}} f(t) e^{-ik\Omega t} dt,
\end{equation}
where for $(a,h)\in \mathcal{X} \times L^p (\mathbb{T})$ we define $(a \odot h)(t) = h(t)a$. It is a standard result in harmonic analysis that for any complex-valued $\varphi \in L^p (\mathbb{T})$, $p\in (1,\infty )$, the sequence of its partial Fourier sums converges to $\varphi$ in $L^p$-norm \cite{Grafakos2009}, and by Carleson-Hunt theorem \cite{MR0238019,Carleson1966}, also pointwise almost everywhere (a.e.) over $\mathbb{T}$. It turns out that these results can be generalized elegantly to the case of Bochner spaces:

\begin{theorem}\label{FourierSeriesConvLp}
If $(\mathcal{X}, \| \cdot \| )$ is a UMD (unconditionality of martingale differences) Banach space with unconditional Schauder basis and $p\in(1,\infty)$, then $(D_n * f)$ converges to $f$ in norm in $\mathscr{L}^p (\mathbb{T},\mathcal{X})$ and also pointwise a.e. to $f(t)$, $t\in\mathbb{T}$, i.e.~for any $f$,
\begin{subequations}
	\begin{equation}
			\lim_{n\to\infty} \int_{\mathbb{T}} \Big\| \sum_{|k|\leqslant n} f_k e^{ik\Omega t} - f(t)\Big\|^{p} dt = 0,
	\end{equation}
	\begin{equation}
			\lim_{n\to\infty} \| (D_n * f)(t) - f(t) \| = 0 \quad \text{for a.e. }t\in\mathbb{T}.
	\end{equation}
\end{subequations}
\end{theorem}

For more detailed version of the above theorem and proofs, see e.g. \cite{ArendtBu2010,Arendt2011} and references therein, as well as \cite{Francia1985}; for UMD spaces, see \cite{9780080532806}.
\par\vspace{\baselineskip}
We say that Fourier series of Banach space-valued function $f : \mathbb{T} \to (\mathcal{X},\| \cdot \|)$ converges \emph{uniformly} to $f$, if $(D_n * f)$ converges to $f$ in $\mathscr{L}^\infty$-norm, i.e.
\begin{equation}
	\lim_{n\to\infty} \esssup_{t\in\mathbb{T}} \| (D_n * f)(t) - f(t)\| = 0 .
\end{equation}
Moreover, Fourier series $\sum_{n\in\integers} f_n \odot e_n$, $f_n \in \mathcal{X}$, will be called \emph{absolutely convergent}, if $(\| f_{n} \|)_{n\in\integers} \in l^1 (\integers)$.
\par\vspace{\baselineskip}
Take $\mathcal{X} = (\matrd,\| \cdot \|)$, where $\| \cdot \|$ is (any) matrix norm on $\matrd$. Then $\mathcal{X}^\dual = (\matrd , \| \cdot \|^\dual)$, where $\| \cdot \|^\dual$ stands for the \emph{matrix dual norm}. If $\| \cdot \|$ is chosen to be the induced operator norm \eqref{eq:InducedOpNorm}, then $\| \cdot \|_{\infty}^{\dual} = \tnorm{\cdot}$, and vice versa \cite{RogerA.Horn2012}. Therefore, ${(\matrd_1)}^\dual \simeq \matrd_\infty = (\matrd, \| \cdot \|_\infty )$ with duality pairing $(a,b)_{\mathbb{C}^{d\times d}} = \tr{ab}$. Space $\matrd_1$ is reflexive, so has a Radon-Nikodym property; therefore $\mathscr{L}^p (\mathbb{T},\matrd_1)^{\dual} \simeq \mathscr{L}^q(\mathbb{T},\matrd_{\infty})$ for $p^{-1} + q^{-1} = 1$, which is complete with respect to norm
\begin{equation}
	\| f \|_{\mathscr{L}^{q}_{\infty}} = \Big( \frac{1}{T} \int_{\mathbb{T}} \| f(t) \|^{2}_{\infty} \, dt \Big)^{1/q}
\end{equation}
and the pairing for $(\mathscr{L}^p,\mathscr{L}^q)$ is $(f,g)_{\mathscr{L}^p} = \frac{1}{T} \int_{\mathbb{T}} \tr{f(t) g(t)} \, dt$.
\par\vspace{\baselineskip}
Now we are ready to formulate a notion of generalized space of states. We propose a following definition within time-independent formalism:

\begin{definition}\label{def:GeneralizedSpace}
The Bochner space $\mathscr{L}^2 (\mathbb{T},\matrd_1)$ of square integrable periodic matrix-valued functions, complete with a norm
\begin{equation}\label{eq:L2BochnerNorm}
	\| f \|_{\mathscr{L}^{2}_{1}} = \Big( \frac{1}{T} \int_{\mathbb{T}} \tnorm{f(t)}^{2} \, dt \Big)^{1/2},
\end{equation}
will be called the generalized space of states, or the Floquet space of states. The corresponding dual space is $\mathscr{L}^2 (\mathbb{T},\matrd_\infty)$ and the pairing is given by
\begin{equation}\label{eq:L2dualityPairing}
	(f,g)_{\mathscr{L}^2} = \frac{1}{T} \int_{\mathbb{T}} \tr{f(t) g(t)} \, dt.
\end{equation}
\end{definition} 

The following three remarks can provide additional justification for the proposed definition of enlarged space of states:

\begin{remark}
We choose the case $p=2$ mainly for good convergence of Fourier series. As $\matrd_1$ is of finite dimension, it is a UMD space; therefore any function $f \in \mathscr{L}^2 (\mathbb{T},\matrd_1)$ possesses a Fourier series which converges in norm and also pointwise a.e. by virtue of Theorem \ref{FourierSeriesConvLp}. This property will be important for existence and convergence of certain Fourier-like expansions of operators later on.
\end{remark}

\begin{remark}
One may define another norm $\| \cdot \|_{\mathscr{L}^{2}_{2}}$ in the space of periodic matrix-valued functions by
\begin{equation}
	\| f \|_{\mathscr{L}^{2}_{2}} = \Big( \frac{1}{T} \int_{\mathbb{T}} \fnorm{f(t)}^{2} \, dt \Big)^{1/2},
\end{equation}
where $\fnorm{a} = (\tr{a^\adj a})^{1/2}$ stands for Frobenius (Hilbert-Schmidt) norm of $a$. Putting $\matrd_2 = (\matrd, \fnorm{\cdot})$, the space $\mathscr{L}^2 (\mathbb{T}, \matrd_2)$ is a Hilbert space. Matrix norms $\tnorm{\cdot}$ and $\fnorm{\cdot}$ are naturally equivalent \cite{RogerA.Horn2012} and it is easy to see $\| \cdot \|_{\mathscr{L}^{2}_{1}}$ and $\| \cdot \|_{\mathscr{L}^{2}_{2}}$ are equivalent as well, thus inducing the same topology on $\mathscr{L}^2 (\mathbb{T},\matrd_1)$. Therefore, spaces $\mathscr{L}^2 (\mathbb{T}, \matrd_1)$ and $\mathscr{L}^2 (\mathbb{T}, \matrd_2)$ are canonically isomorphic as Banach spaces. This identification may be useful for various computational reasons (we will not, however, exploit it in this article).
\end{remark}

\begin{remark}
We choose $\tnorm{\cdot}$ norm in the target space for two reasons. First, it is a natural choice of norm of density matrices in $\matrd$ in quantum theory and is the most commonly used one, as it is numerically equal to the \emph{trace} of any positive semi-definite matrix. Second, it allows certain evolution maps on the enlarged space of states to be \emph{contractions} with respect to this norm, which is an expected property of irreversible quantum evolution.
\end{remark}

\subsection{Isomorphic representations of generalized space of states}
\label{sec:IsomorphicRep}

Bochner spaces may be given a standard \emph{tensorial} representation in a following manner: the algebraic tensor product $\matr{d} \otimes  L^2 (\mathbb{T})$ is embedded in $\mathscr{L}^2 (\mathbb{T},\matrd_1)$ via a natural injection $\iota$ defined on simple tensors as $\iota (a\otimes h) = a\odot h$ and then extended by linearity. This allows to endow this tensor product with a norm $x \mapsto \| x \|_{\mathscr{L}^{2}_{1}} = \| \iota(x) \|_{\mathscr{L}^{2}_{1}}$ (we use the same symbol to denote both norms). In result, the injection extends to an isometric isomorphism
\begin{equation}
	\iota : \matrd \, \bar{\otimes} \, L^2 (\mathbb{T}) \to \mathscr{L}^{2}(\mathbb{T},\matrd_1),
\end{equation}
where $\matrd \, \bar{\otimes} \, L^2 (\mathbb{T})$ is the completion of algebraic tensor product with respect to $\| \cdot \|_{\mathscr{L}^{2}_{1}}$; we will use a notation $\tilde{f}$ for elements in $\matrd \, \bar{\otimes} \, L^2 (\mathbb{T})$ to indicate the isomorphism, i.e.~$f = \iota(\tilde{f})$. We note that $\| \cdot \|_{\mathscr{L}^{2}_{1}}$ is not a tensor norm, as it does not satisfy the so-called mapping property \cite{A.Defant1992} and alternative choices are possible \cite{Calabuig2015}.

For each $\tilde{f} \in \matrd \,\bar{\otimes}\, L^2 (\mathbb{T})$ one can find a unique sequence $(f_n)_{n\in\integers} \subset \matrd$ such that $\tilde{f} = \sum_{n\in\integers} f_n \otimes e_n$, $f = \sum_{n\in\integers} f_n \odot e_n$, both series norm-convergent. Clearly, $\sum_{n\in\integers} f_n \odot e_n$ is a Fourier series of function $f = \iota(\tilde{f})$ and, since $\matrd_1$ is a UMD space, it converges pointwise a.e. by theorem \ref{FourierSeriesConvLp},
\begin{equation}\label{eq:MatrixFourierSeries}
	f(t) \stackrel{\mathrm{a.e.}}{=}\sum_{n\in\integers} f_n e^{in\Omega t}, \qquad f_n = \frac{1}{T} \int_{\mathbb{T}} f(t) e^{-in\Omega t} dt.
\end{equation}
Duality pairing in $\matrd \,\bar{\otimes}\, L^2 (\mathbb{T})$ will be then given by $(\tilde{f},\tilde{g})_{\bar{\otimes}} = (f,g)_{\mathscr{L}^2}$. The Bochner space and its tensorial form are identified, and so they can be referred to as \emph{the generalized space of states} simultaneously.

\begin{remark}\label{remark:Comparison}
We remark here, that despite the fact that the above construction of \emph{tensorial} form of enlarged space corresponds closely to the traditional one as given by Shirley, Sambe and Howland \cite{Shirley_1965,Sambe_1973,Howland1974,Howland1979}, construction presented in this article is notably different. In original approach, the generalized space was defined as a topological Hilbertian tensor product $\mathcal{H} \hat{\otimes} L^2 (\mathbb{T}) \simeq \mathscr{L}^2 (\mathbb{T},\mathcal{H})$, where $\mathcal{H}$ stands for a Hilbert space of (pure) states of the system and $\hat{\otimes}$ is a Hilbert space completion with respect to induced cross norm. As the traditional approach concerned unitary evolution of closed systems, such Hilbert space based construction was perfectly feasible. In our approach, the generalized space is no longer a Hilbert space (unless $p=2$ and the target space is Hilbert, which is not the case here), but rather a topological tensor product of two Banach spaces, complete with respect to Bochner space norm $\| \cdot \|_{\mathscr{L}^{2}_{1}}$ \eqref{eq:L2BochnerNorm}, which is not a tensor norm.
\end{remark}

\subsubsection{Operators on generalized space of states}
\label{subsubsect:OperatorsOnS}

The Banach algebra $B(\matrd_1)$ of all linear maps on $\matrd_1$ may be identified with $\complexes^{d^2 \times d^2}$, $\dim{B(\matrd_1)=d^4}$. We endow $B(\matrd_1)$ with supremum norm
\begin{equation}
	\| A \|_{\infty} = \sup_{\tnorm{a}\leqslant 1} \tnorm{A(a)}, \qquad a\in\matrd_1.
\end{equation}

\begin{definition}\label{def:FourierShifts}
Fourier shift operators $F_{n}$, $n\in\integers$ and Fourier number operator $F_z$, all acting on $L^{2}(\mathbb{T})$, are defined via equalities
\begin{equation}
	F_{n}(e_{m}) = e_{m+n}, \quad F_z (e_n) = n e_n , \qquad m,n\in\integers .
\end{equation}
\end{definition}

\begin{proposition}\label{prop:FourierOpComm}
Fourier operators have the following properties:
\begin{enumerate}
	\item\label{prop:FourierOpCommOne} $\{F_n : n\in\integers\}$ is a unitary commutative representation of group $(\integers, +)$,
	\item\label{prop:FourierOpCommTwo} $\comm{F_z}{F_n} = n F_n$,
	\item\label{prop:FourierOpCommThree} $F_z$ is self-adjoint and unbounded.
\end{enumerate}
\end{proposition}

\begin{proof}
Properties (\ref{prop:FourierOpCommOne}), (\ref{prop:FourierOpCommTwo}) as well as self-adjointness of $F_z$ can be easily shown by direct computation. Unboundedness of $F_z$ is also clear: simply consider sequence of basis vectors $(e_{n})_{n\in\integers}$ and notice $\sup_{n\in\integers}\| F_z (e_{n})\|_{L^2} = \sup_{n\in\integers} |n| = \infty$.
\end{proof}

Property (\ref{prop:FourierOpCommOne}) leads in particular to group-like conditions $F_0 = I$, $F_n = F_{1}^{n}$, $F_{n+m} = F_{n}F_{m}$, $F_{n}^{-1} = F_{-n} = F_{n}^{\hadj}$ and $\|F_n\|_{\infty}=1$. The idea standing behind introduction of Fourier operators is such that they may be efficiently used to represent periodic operator-valued functions as time-independent \emph{static} operators defined on the generalized space of states. This again will remain compatible with Howland formulation.

Let $A$ be a linear operator on $\mathscr{L}^{2}(\mathbb{T},\matrd_1)$. Then, the unique operator $\tilde{A}$ acting on $\matrd_1 \, \bar{\otimes} \, L^2 (\mathbb{T})$ given by
\begin{align}\label{eq:FourierLiftingDefinition}
	\tilde{A} = \iota^{-1}\circ A\circ\iota, \quad \domain{\tilde{A}} = \iota^{-1}(\domain{A}),
\end{align}
will be called the \emph{Fourier lifting of} $A$. Naturally, mapping $A \mapsto \tilde{A}$ is a bijection. By making a proper choice of $A$ one can then uniformly express operators, possibly time-dependent, acting on matrix-valued functions, as ``static'' operators acting on (a subspace of) $\matrd_1 \, \bar{\otimes} \, L^2 (\mathbb{T})$. In realm of time-independent formalism, this operation allows for passing from ODE-based analysis to time-independent, purely algebraic one. By additionally assuming denseness of $\domain{A}$, one may ensure a definition of Fourier lifting to be well-suited for representing discontinuous operators (like a derivative operator). Namely, if $\domain{A}$ is dense in $\mathscr{L}^{2}(\mathbb{T},\matrd_1)$, one can then define its dual (adjoint) operator $A^\dual$ by imposing
\begin{align}\label{eq:L2dualMap}
	(f, A(g))_{\mathscr{L}^2} &= \frac{1}{T} \int_{\mathbb{T}} \tr{f(t) (A\circ g)(t)} dt \\
	&= \frac{1}{T} \int_{\mathbb{T}} \tr{(A^\dual\circ f)(t) g(t)} dt = (A^\dual (f),g)_{\mathscr{L}^2} \nonumber
\end{align}
for duality pairing \eqref{eq:L2dualityPairing}. Its domain $\domain{A^\dual}$ consists of such $\varphi \in \mathscr{L}^2 (\mathbb{T},\matrd_\infty)$ that $\varphi \circ A$ extends boundedly on entire $\mathscr{L}^2 (\mathbb{T},\matrd_1)$; similarly, $\tilde{A}^\dual$ exists and it is straightforward to show, that $\tilde{A}^\dual = \iota^{-1}\circ A^\dual \circ \iota$ and $\domain{\tilde{A}^\dual} = \iota^{-1}(\domain{A^\dual})$, i.e.~$\tilde{A}^\dual$ is also a Fourier lifting of $A^\dual$.

Denote by $\mathscr{C}^{k}(\mathbb{T},\matrd_1)$, $k \in \{0,1,\infty \}$ the linear spaces of continuous ($k=0$), continuously differentiable ($k=1$) and smooth ($k=\infty$) matrix-valued functions on $\mathbb{T}$. As each periodic continuous function is bounded, in particular we have $\mathscr{C}^0 (\mathbb{T},\matrd_1) \subset \mathscr{L}^{2}(\mathbb{T},\matrd_1)$ set-theoretically. $\mathscr{C}^1 (\mathbb{T},\matrd_1)$ is then a dense linear subspace of $\mathscr{L}^{2}(\mathbb{T},\matrd_1)$, as it contains all the smooth functions.

\subsubsection{Some algebraic properties}

For $f,g \in \mathscr{L}^{2}(\mathbb{T},\matrd_1)$, we define their product $fg$ pointwise as $(fg)(t) = f(t) g(t)$ and for $\tilde{f},\tilde{g} \in \matrd_1 \, \bar{\otimes} \, L^2 (\mathbb{T})$ we have $\tilde{f}\tilde{g} = \iota^{-1}(fg)$. The target space of such multiplication can be unfortunately much larger than original space, as $\mathscr{L}^{2}(\mathbb{T},\matrd_1)$ is not an algebra. Nevertheless, multiplication by any essentially bounded function is well-posed, i.e.~inclusions $\mathscr{L}^\infty \mathscr{L}^{2}, \, \mathscr{L}^{2} \mathscr{L}^\infty\subset \mathscr{L}^{2}$ hold set-theoretically; to see this, simply estimate for, say, $f\in\mathscr{L}^{\infty}$ and $g\in\mathscr{L}^2$,
\begin{equation}
	\|fg\|^{2}_{\mathscr{L}^{2}_{1}} \leqslant \esssup_{t\in\mathbb{T}}{\tnorm{f(t)}^{2}} \cdot \frac{1}{T}\int_{\mathbb{T}} \tnorm{g(t)}^{2} \, dt = \| f \|_{\mathscr{L}^{\infty}}^{2} \| g \|_{\mathscr{L}^{2}}^{2}
\end{equation}
which easily comes from inequality \eqref{eq:InequalityOfMatrixNorms}. As a result, $fg \in \mathscr{L}^{2}(\mathbb{T},\matrd_1)$. Picking $g\in\mathscr{L}^{\infty}(\mathbb{T},\matrd_1)$ yields the second inclusion. It is easy to check that constant function $I\odot e_0$ is a neutral element of multiplication in $\mathscr{L}^{2}(\mathbb{T},\matrd_1)$. Likewise, $I \otimes e_0 = \iota^{-1}(I\odot e_0)$ is a neutral element of multiplication in $\matrd_1 \, \bar{\otimes} \, L^2 (\mathbb{T})$. 

Let $h \mapsto h^\adj$, $h^\adj (t) = \overline{h(t)}$ be an involution on $L^2 (\mathbb{T})$. A conjugate-linear injective map $f \mapsto f^\adj$ given on $\mathscr{L}^{2}(\mathbb{T},\matrd_1)$ as
\begin{equation}
	f^\adj (t) = f(t)^\hadj , \quad f^\adj = \sum_{n\in\integers} f_{n}^{\hadj} \odot e_{n}^{\adj},
\end{equation}
is then naturally an \emph{involution} on $\mathscr{L}^{2}(\mathbb{T},\matrd_1)$; we will call $f^\adj$ the \emph{adjoint} of $f$. Similarly, we define adjoint on $\matrd_1 \, \bar{\otimes} \, L^2 (\mathbb{T})$ by imposing $\tilde{f}^\adj = \iota^{-1}(f^\adj) = \sum_{n\in\integers}f_{n}^{\hadj}\otimes e_{n}^{\adj}$.

\subsubsection{Complete positivity, trace preservation and contractions}

Dynamical maps posses two important mathematical properties -- the \emph{complete positivity} and \emph{trace preservation} -- imposed to ensure the overall ``physicality'' of the evolution. One can then ask if it is reasonable to expect similar conditions to be satisfied by generalized quantum dynamics, lifted to space $\mathscr{L}^2 (\mathbb{T},\matrd_1)$. As we show later on, this indeed is the case: the generalized dynamics will exhibit complete positivity and trace preservation (as well as contractivity) conditions, in the sense which we define in this section. What is important, such situation perfectly resembles the original Hamiltonian time-independent formalism as proposed by Howland -- in that case, the evolution maps acting on Howland enlarged space of states inherited their properties, like unitarity, from ordinary time-dependent evolution operators acting on regular state space of the system. In our formulation, we observe the very same correspondence between maps on $\matrd_1$ and $\mathscr{L}^{2}(\mathbb{T},\matrd_1)$ as they also share the most crucial and expected properties.
\par\vspace{\baselineskip}
We will call $f\in\mathscr{L}^2 (\mathbb{T},\matrd_1)$ \emph{positive}, $f \geqslant 0$, if and only if $f(t) \geqslant 0$ for every $t\in\mathbb{T}$. The positive cone $P^{+}$ will be then generated by all elements of a form $g^\adj g$, i.e.~$P^+ = \Big\{ \sum_{j=1}^{n} g_{j}^{\adj} g_{j} : g_j \in \mathscr{L}^{2}, \, n\in\naturals\Big\}$.

For given operator $A$, let $A_n = I \otimes A$ denote its extension on tensor product space $\complexes^{n\times n} \otimes \mathscr{L}^{2}(\mathbb{T},\matrd_1)$. This space is naturally identified with $M_n (\mathscr{L}^2)$ i.e.~every $\mathbf{f} \in \complexes^{n\times n} \otimes \mathscr{L}^{2}$ is uniquely represented as a matrix $\mathbf{f} = [f_{ij}]$ such that $f_{ij} \in \mathscr{L}^2(\mathbb{T},\matrd_1)$. The adjoint $\mathbf{f}^\adj$ is naturally expressed as $\mathbf{f}^\adj = [f_{ji}^{\adj}]$ and the positive cone $P_{n}^{+} \subset M_n (\mathscr{L}^2)$ is $P_{n}^{+} = \{ \sum_{j=1}^{m} \mathbf{g}_{j}^{\adj} \mathbf{g}_j : m < \infty \}$. The action of $A_n$ on $P_{n}^{+}$ is then simply $A_{n}([f_{ij}]) = [A(f_{ij})]$.

\begin{definition}\label{eq:DefinitionNpos}
Densely defined operator $A$, $\domain{A} \subset \mathscr{L}^2 (\mathbb{T},\matrd_1)$, will be called:
\begin{enumerate}
	\item positive, $A \geqslant 0$, iff $A(\domain{A} \cap P^{+}) \subset P^{+}$,
	\item $n$-positive, $A\in\mathrm{P}_{n}(\mathscr{L}^2)$, iff $A_n \geqslant 0$, i.e.~$A_n (\domain{A_n} \cap P_{n}^{+}) \subset P_{n}^{+}$,
	\item completely positive, $A\in\cp{\mathscr{L}^2}$, iff $A\in\mathrm{P}_{n}(\mathscr{L}^2)$ for all $n\in\naturals$.
\end{enumerate}
\end{definition}

Complete positivity and $n$-positivity on $\matrd_1 \, \bar{\otimes} \, L^2 (\mathbb{T})$ will be then defined and denoted analogously as $\cp{\bar{\otimes}}$ and $\mathrm{P}_{n}(\bar{\otimes})$, respectively. The correspondence between positivity on both spaces is also evident, i.e.~$A\in\mathrm{P}_{n}(\mathscr{L}^{2})$ iff $\tilde{A}\in\mathrm{P}_{n}(\bar{\otimes})$ and $A\in\cp{\mathscr{L}^{2}}$ iff $\tilde{A}\in\cp{\bar{\otimes}}$.
\par\vspace{\baselineskip}
Positivity preserving condition, imposed on quantum dynamics is almost always paired together with \emph{trace preserving} condition, important for pertaining a statistical interpretation of density operators as quantum-mechanical probability distributions. A specifically stated definition of trace on the enlarged space will allow to extend this condition onto generalized dynamics, too. Canonical \emph{trace} of a matrix is a positive linear functional on $\matrd$, uniquely represented by identity matrix such that $\tr{a} = (I,a)_{\matrd}$ for any $a \in \matrd$. One can then formulate similar definition of traces on both $\mathscr{L}^2 (\mathbb{T},\matrd_1)$ and $\matrd_1 \, \bar{\otimes} \, L^2 (\mathbb{T})$:

\begin{definition}\label{def:TraceOnGenSpace}
Linear functional $\trl{} : \mathscr{L}^2 (\mathbb{T},\matrd_1) \to \complexes$ represented by neutral element in $\mathscr{L}^2 (\mathbb{T},\matrd_1)$, i.e.~given by
\begin{equation}\label{eq:TraceL2definition}
	\trl{f} = (I \odot e_0,f)_{\mathscr{L}^2} = \frac{1}{T}\int_{\mathbb{T}} \tr{f(t)} \, dt,
\end{equation}
will be called a trace on $\mathscr{L}^2 (\mathbb{T},\matrd_1)$. By analogy, the trace on $\matrd_1 \, \bar{\otimes} \, L^2 (\mathbb{T})$ will be defined as
\begin{equation}\label{eq:TraceSdefinition}
	\trs{\tilde{f}} = (I\otimes e_0,\tilde{f})_{\bar{\otimes}} = (I \odot e_0 , f)_{\mathscr{L}^2}.
\end{equation}
\end{definition}

Let us denote sets of all trace preserving maps on $\mathscr{L}^2 (\mathbb{T},\matrd_1)$ and $\matrd_1 \, \bar{\otimes} \, L^2 (\mathbb{T})$ by $\tp{\mathscr{L}^2}$ and $\tp{\bar{\otimes}}$, respectively. One can immediately show that we have $A \in \tp{\mathscr{L}^2}$ iff $\tilde{A}\in\tp{\bar{\otimes}}$. A following simple, yet useful lemma then follows (we leave a proof for the reader):

\begin{lemma}\label{prop:TPneutralElementPreservation}
$A \in \tp{\mathscr{L}^2}$ iff $A^\dual$ is unital and $\tilde{A}\in\tp{\bar{\otimes}}$ iff $\tilde{A}^\dual$ is unital.
\end{lemma}

\subsection{Examples of operators and their properties}
\label{sec:ExamplesOfOperators}

Here we will construct and provide some basic properties for three important operators on space $\mathscr{L}^{2}(\mathbb{T},\matrd_1)$, together with their counterparts acting on $\matrd_1 \, \bar{\otimes} \, L^2 (\mathbb{T})$; this will include the \emph{time-dependent operator-valued function}, \emph{derivative operator} and \emph{shift operator}. The next section will provide some concrete expressions for aforementioned operators; in particular, the infinite series representation will be introduced (with some certain convergence issues addressed).

\subsubsection{Operator-valued function}
\label{sect:Operatorvaluedfunction}

Let $U \in B(\matrd_1)$. Its dual $U^\dual$ is then given via duality pairing on $\matrd_1$ by
\begin{equation}\label{eq:XadjointDefinition}
	(a,U(b))_{\matrd} = \tr{a \, U (b)} = \tr{U^{\dual} (a) \, b} = (U^{\dual}(a),b)_{\matrd}.
\end{equation}
Now, let $t\mapsto A_t \in B(\matrd_1)$ be a periodic, operator-valued function. It induces a linear operator $A$ acting on $\mathscr{L}^{2}(\mathbb{T},\matrd_1)$ via
\begin{equation}\label{eq:AlphaFt}
	f \mapsto (A \circ f)(t) = A_t (f(t)), \qquad f \in \mathscr{L}^{2}(\mathbb{T},\matrd_1), \, t \in \mathbb{T}.
\end{equation}
In all the following, we will be assuming $t\mapsto A_t$ is a bounded function of $t$, i.e.~$\sup_{t\in\mathbb{T}} \| A_t \|_{\infty} = C$ is finite, where $\| A_t \|_\infty = \sup_{\tnorm{a}\leqslant 1} \tnorm{A_t(a)}$ is the operator norm of $A_t$ as a map on $\matrd_1$. In such case it is easy to notice
\begin{equation}
		\frac{1}{T}\int_{\mathbb{T}} \tnorm{A_t(f(t))}^{2} \, dt \leqslant C^2 \|f\|_{\mathscr{L}^{2}_{1}}^{2}
\end{equation}
for every $f\in\mathscr{L}^{2}(\mathbb{T},\matrd_1)$, so $A$ is a bounded endomorphism on $\mathscr{L}^{2}(\mathbb{T},\matrd_1)$. Employing \eqref{eq:L2dualMap} we find $A^\dual$ given by
\begin{equation}\label{eq:AlphaFboundedFadj}
	A^{\dual}(f)(t) = A_{t}^{\dual}(f(t)),
\end{equation}
where the prime symbol at the right hand side denotes a dual in a sense of \eqref{eq:XadjointDefinition}.
Its Fourier lifting $\tilde{A}^{\dual}$ then satisfies, due to \eqref{eq:FourierLiftingDefinition},
\begin{equation}\label{eq:AlphaFbounded}
	(\tilde{A}^{\dual}\circ \iota)(\tilde{x})(t) = A_{t}^{\dual}(x(t)).
\end{equation}

\begin{proposition}\label{thm:AlphaTP}
$A \in \tp{\mathscr{L}^2}$ if and only if $A_t \in \tp{\matrd}$ for all $t\in\mathbb{T}$.
\end{proposition}

\begin{proof}
By Lemma \ref{prop:TPneutralElementPreservation}, $A^\dual$ is unital and therefore
\begin{equation}
	A^\dual (I\odot e_0)(t) = A_{t}^{\dual}(I) = I
\end{equation}
and $A_t$ is unital as well; as for any $a\in\matrd_1$ we have
\begin{equation}\label{eq:AlphaTPproof}
	\tr{a} = (I,a)_{\matrd} = (A_{t}^{\dual}(I),a)_{\matrd} = (I,A_{t}(a))_{\matrd} = \tr{A_{t}(a)},
\end{equation}
automatically $A_t \in \tp{\matrd_1}$. On the contrary, supposing $A_t$ trace preserving and applying \eqref{eq:TraceSdefinition} we immediately obtain the opposite.
\end{proof}

\begin{proposition}\label{prop:AlphaContractive}
If $A_t \in B(\matrd_1)$ is a trace norm contraction on $\matrd_1$ for each $t\in\mathbb{T}$, then $A$ is a contraction on $\mathscr{L}^{2}(\mathbb{T},\matrd_1)$.
\end{proposition}

\begin{proof}
This claim is shown by simple computation
\begin{equation}
	\| A (f) \|_{\mathscr{L}^{2}_{1}}^{2} = \frac{1}{T} \int_{\mathbb{T}} \tnorm{A_t (f(t))}^{2} \, dt \leqslant \frac{1}{T} \int_{\mathbb{T}} \tnorm{f(t)}^{2} \, dt = \| f \|_{\mathscr{L}^{2}_{1}}^{2},
\end{equation}
as $\tnorm{A_t (a)} \leqslant \tnorm{a}$.
\end{proof}

\begin{remark}
The opposite claim does \textbf{not} hold. As a counterexample, take $A$ defined as a multiplication operator via $A (f) = \xi f$, where $\xi \in L^2 (\mathbb{T})$, $f(t) \geqslant 0$, $\| f \|_{L^2} = 1$, and $f(t) > 1$ for $t$ in a finite family of sub-intervals in $[0,T)$. Take any $f = a \odot e_0$, i.e.~a constant function $f(t) = a \in \matrd$; one has, after simple algebra,
\begin{equation}
	\| A (a \odot e_0) \|_{\mathscr{L}^{2}} = \Big(\frac{1}{T} \int_{\mathbb{T}} \tnorm{A_t (a)} \, dt \Big)^{1/2} = \tnorm{a} \| \xi \|_{L^2} = \tnorm{a},
\end{equation}
so $A$ is a contraction. However, take $\xi$ as, say, quadratic function of $t$ with zeroes at points $t=0$ and $t=T$, given as
\begin{equation}
	\xi (t) = -\frac{\sqrt{30}}{T^2} (t-T)t.
\end{equation}
Then, one easily finds $\| \xi \|_{L^2} = 1$ and $\xi(t) > 1$ over interval $(t_1, t_2)$ centered around $T/2$, where the maximal point of $\xi$ is located, $\xi (T/2) = \sqrt{30}/4 > 1$. For each $t\in (t_1, t_2)$ we therefore have $\tnorm{A_t (a)} = \xi(t) \tnorm{a} > \tnorm{a}$, and so $A_t$ is not a trace norm contraction on $\matrd_1$.
\end{remark}

\begin{proposition}\label{thm:AlphaAtCP}
The following hold:
\begin{enumerate}
	\item\label{point:AlphaAtCPone} Let $A_t \in \cp{\matrd_1}$ for every $t\in\mathbb{T}$. If additionally $A_t$ admits a Kraus representation
	\begin{equation}\label{eq:AlphaAtCP}
		A_t (a) = \sum_{j=1}^{d^2} X_{j,t}^{\hadj} a X_{j,t}, \qquad a\in\matrd_1,
	\end{equation}
	such that for each $j$, a mapping $t\mapsto X_{j,t}$ is a bounded matrix-valued function of $t$, then $A \in \cp{\mathscr{L}^2}$.
	\item\label{point:AlphaAtCPtwo} If $A \in \cp{\mathscr{L}^2}$, then $A_t \in \cp{\matrd_1}$ for every $t\in\mathbb{T}$.
\end{enumerate}
\end{proposition}

\begin{proof}
Ad (\ref{point:AlphaAtCPone}). Let $A_t$ be given by \eqref{eq:AlphaAtCP} and take a sequence of bounded functions $(X_j)$ such that $X_j (t) = X_{j,t}$. Then, for each $f\in\mathscr{L}^2(\mathbb{T},\matrd_1)$, we also have $X_j ^\adj f X_j \in \mathscr{L}^2(\mathbb{T},\matrd_1)$ and
\begin{equation}
	A_t (f(t)) = \Big( \sum_{j=1}^{d^2} X_{j}^{\adj} f X_{j} \Big) (t) \,\, \Rightarrow \,\, A(f) = \sum_{j=1}^{d^2} X_{j}^{\adj} f X_{j}.
\end{equation}
Take any $n \geqslant 1$ and any $\mathbf{f} = [f_{ij}]_{n\times n}$, $f_{ij}\in\mathscr{L}^2$. For $A_n = I \otimes A$, we have
\begin{align}
	A_n (\mathbf{f}^\adj \mathbf{f}) &= (I \otimes A)([f_{ij}]_{n\times n}^{\adj}[f_{ij}]_{n\times n}) = \Big[\sum_{k=1}^{n} A(f_{ki}^{\adj} f_{kj})\Big]_{n\times n} \\
	&= \Big[\sum_{l=1}^{d^2}\sum_{k=1}^{n} X_{l}^{\adj} f_{ki}^{\adj} f_{kj} X_{l} \Big]_{n\times n} = \sum_{l=1}^{d^2} \Big[\sum_{k=1}^{n} (f_{ki}X_l)^{\adj} f_{kj}X_l \Big]_{n\times n} \nonumber \\
	&= \sum_{l=1}^{d^2} [f_{ij}X_l]_{n\times n}^{\adj} [f_{ij}X_l]_{n\times n} \in P_{n}^{+},
\end{align}
so $A \in \mathrm{P}_{n}(\mathscr{L}^2)$. As the above remains valid for any $n\geqslant 1$, $A \in \cp{\mathscr{L}^2}$.

Ad (\ref{point:AlphaAtCPtwo}). Take $A \in \cp{\mathscr{L}^2}$. Then $A_n \in \mathrm{P}_{n}(\mathscr{L}^2)$ for every $n\geqslant 1$, i.e.
\begin{equation}\label{eq:AnNPositivity}
	A_n (\mathbf{f}^\adj \mathbf{f}) = \Big[\sum_{k=1}^{n} A(f_{ki}^{\adj} f_{kj})\Big]_{n\times n} = \sum_{k=1}^{m} \mathbf{g}_{k}^{\adj} \mathbf{g}_{k} \in P_{n}^{+},
\end{equation}
for some elements $\mathbf{g}_k = [g^{(k)}_{ij}]_{n\times n}$ and $m < \infty$. Note, that $\complexes^{n\times n} \otimes \mathscr{L}^{2} \simeq M_n(\mathscr{L}^{2})$ may be injectively embedded in space of all periodic functions with values in $M_d (\complexes^{n\times n}) \simeq \complexes^{nd\times nd}$ via the evaluation map $\eta$, defined as $\eta (\mathbf{f})(t) = \mathbf{f}(t) = [f_{ij}(t)]_{n\times n}$, $f_{ij}(t) \in \matrd$. Then, acting with $\eta$ on \eqref{eq:AnNPositivity} we obtain
\begin{equation}\label{eq:AnNPositivity2}
	\Big[\sum_{k=1}^{n} A_t (f_{ki}(t)^{\adj} f_{kj}(t))\Big]_{n\times n} = \sum_{k=1}^{m} \mathbf{g}_{k}(t)^{\adj} \mathbf{g}_{k}(t),
\end{equation}
where $\mathbf{g}_{k}(t) = [g_{ij}^{(k)}(t)]_{n\times n}$, $g_{ij}^{(k)}(t) \in \matrd$. Left hand side of \eqref{eq:AnNPositivity2} may be however put in a form
\begin{equation}
	\Big[\sum_{k=1}^{n} A_t (f_{ki}(t)^{\hadj} f_{kj}(t))\Big]_{n\times n} = (I \otimes A_t)([f_{ij}(t)]_{n\times n}^{\adj} [f_{ij}(t)]_{n\times n}),
\end{equation}
which yields
\begin{equation}
	(I \otimes A_t ) (\mathbf{f}(t)^\adj \mathbf{f}(t)) = \sum_{k=1}^{m} \mathbf{g}_{k}(t)^{\adj} \mathbf{g}_{k}(t).
\end{equation}
In particular, one can take $\mathbf{f}$ to be any constant matrix-valued function, i.e.~$\mathbf{f}(t) = a \in\complexes^{nd\times nd}$. As space of all constant functions is isomorphic to entire $\complexes^{nd\times nd}$, we have for every $a\in\complexes^{nd\times nd}$,
\begin{equation}
	(I \otimes A_t)(a^\adj a) = \sum_{k=1}^{m} \mathbf{g}_{k}(t)^{\adj} \mathbf{g}_{k}(t)
\end{equation}
for some $(\mathbf{g}_k)$, i.e.~$A_t$ is $n$-positive for every $t\in\mathbb{T}$, and since the above result does not depend on $n$, also completely positive.
\end{proof}

\subsubsection{Time derivative and right shift operators}

Each $f\in\mathscr{L}^2(\mathbb{T},\matrd_1)$ may be expressed as a matrix of functions $f(t) = [f_{kl}(t)]_{d\times d}$. Then, the derivative operator $\partial$ on $\mathscr{L}^2(\mathbb{T},\matrd_1)$ will be defined, by Lemma \ref{lemma:FrechetDiffEquivMatElementsDiff}, in standard manner as Fr\'{e}chet derivative
\begin{equation}
	\partial(f)(t) = \Big[\frac{df_{kl}(t)}{dt}\Big]_{d\times d}, \quad \domain{\partial} = \mathscr{C}^{1}(\mathbb{T},\matrd_1) .
\end{equation}
$\partial$ is densely defined (as space $\mathscr{C}^1 (\mathbb{T},\matrd_1)$, containing all smooth functions, is dense in $\mathscr{L}^2(\mathbb{T},\matrd_1)$), unbounded and closed. By isometry, its Fourier lifting $\tilde{\partial}$ is also densely defined on $\domain{\tilde{\partial}} = \iota^{-1}(\mathscr{C}^1 (\mathbb{T},\matrd_1))$, unbounded and closed.

\begin{proposition}\label{thm:SrepresentationDelta}
Fourier lifting of $\partial$ and its dual $\tilde{\partial}$ admit expressions
\begin{equation}
	\tilde{\partial} = i\Omega \, I \otimes F_z, \qquad \tilde{\partial}^{\dual} = -\tilde{\partial} .
\end{equation}
\end{proposition}

\begin{proof}
Let us take any function $f\in\mathscr{L}^2(\mathbb{T},\matrd_1)$ and set $\xi_n = \sum_{|k|\leqslant n} f_k \odot e_k$, $f_k \in \matrd_1$, to be its partial Fourier series. As $\xi_n$ is clearly differentiable, we have
\begin{equation}
	\partial(\xi)(t) = i\Omega \sum_{|k|\leqslant n} k \, f_k e^{ik\Omega t} = \Big( \sum_{|k|\leqslant n} i\Omega \, f_k \odot F_z (e_k)\Big)(t)
\end{equation}
Take $\tilde{\xi}_n = \iota^{-1}(\xi_n) = \sum_{|k|\leqslant n} f_k \otimes e_k \in \matrd_1 \bar{\otimes} L^2(\mathbb{T})$. As $\iota^{-1} (f_k \odot F_z (e_k)) = f_k \otimes F_z(e_k)$ we immediately have
\begin{equation}
	\tilde{\partial}(\tilde{\xi}_n) = (\iota^{-1}\circ\partial\circ\iota)(\tilde{\xi}_n) = i\Omega\sum_{|k|\leqslant n} f_k \otimes F_z (e_k) = (i\Omega \, I \otimes F_z)(\tilde{\xi}_n).
\end{equation}
Since $\xi_n \to f$ and $\partial$ is closed, $\partial(\xi_n)\to\partial(f)$; similarly $\tilde{\xi}_n \to \tilde{f} = \iota^{-1}(f)$ and $(i\Omega \, I \otimes F_z)(\tilde{\xi}_n)\to (i\Omega \, I \otimes F_z)(\tilde{f})$. This yields $\tilde{\partial} = i\Omega \, I \otimes F_z$, as claimed.

For $f,g\in\mathscr{C}^1 (\mathbb{T},\matrd_1)$ it is easy to show $(f,\partial(g))_{\mathscr{L}^2} = (-\partial (f), g)_{\mathscr{L}^2}$ and so $\partial^\dual = -\partial$, which comes from integrating by parts over $\mathbb{T}$ (see \cite{Treves1967} for details); as a result, $\tilde{\partial}^\dual = -\tilde{\partial}$.
\end{proof}

Finally, we address the last example, namely the shift operator. The unitary abelian group $\{S_{\tau}\}$ of shift operators on $L^2 (\mathbb{T})$, acting by $S_{\tau}(g)(t) = g(t+\tau)$ may be easily shown to be generated by $i\Omega F_z$: indeed, for any $g \in \mathcal{C}^1 (\mathbb{T}) \subset L^2 (\mathbb{T})$ simply compute
\begin{equation}
	i\Omega \, F_z (g)(t) = i\Omega \sum_{n\in\integers} n g_n e^{in\Omega t} =\frac{dg(t)}{dt},
\end{equation}
so $S_{\tau} = \exp{\tau\frac{d}{dt}} = \exp{i\tau\Omega F_z}$. Let us denote its extension onto $\mathscr{L}^2(\mathbb{T},\matrd_1)$ also by $S_\tau$, i.e.~$S_{\tau}(f)(t) = f(t+\tau)$. Analogously, define $\Delta_\tau$ as a \emph{right shift operator} on $\mathscr{L}^{2}(\mathbb{T},\matrd_1)$ by imposing
\begin{equation}\label{eq:RightShiftDefinition}
	\Delta_{\tau}(f)(t) = S_{-\tau}(f)(t) = f(t-\tau) , \quad f\in \mathscr{L}^{2}(\mathbb{T},\matrd), \,\tau \in [0,\infty )
\end{equation}
and denote by $\tilde{\Delta}$ its Fourier lifting.

\begin{proposition}\label{prop:DeltaTauSrep}
We have $\tilde{\Delta}_\tau = I \otimes e^{-i\tau\Omega F_z}$ and set $\{\tilde{\Delta}_\tau : \tau \geqslant 0\}$ is a semigroup of right shift operators on $\mathscr{L}^{2}(\mathbb{T},\matrd_1)$, generated by $-\tilde{\partial} = -i\Omega \, I \otimes F_z$.
\end{proposition}

\begin{proof}
For any $\tilde{f}\in\matrd_1 \bar{\otimes} L^2(\mathbb{T})$,
\begin{align}
	(\iota \circ I &\otimes e^{-i\tau \Omega F_z})(\tilde{f}) = \sum_{n\in\integers} \iota(f_n \otimes e^{-i\tau\Omega F_z}(e_n)) \\
	&= \sum_{n\in\integers} f_n \odot e^{-i\tau\Omega F_z}(e_n) = \sum_{n\in\integers} f_n \odot S_{-\tau}(e_n) = S_{-\tau}(f),\nonumber
\end{align}
so indeed $\tilde{\Delta}_{\tau} = I \otimes e^{-i\tau\Omega F_z}$ is the Fourier lifting of right shift $\Delta_\tau$ and, by simple algebra, the semigroup properties are obvious. Computing the derivative of function $\tau \mapsto \tilde{\Delta}_\tau$ automatically proves the second claim, i.e.~$-\tilde{\partial} = -i\Omega \, I \otimes F_z$ indeed generates the semigroup $\{\tilde{\Delta}_\tau\}$.
\end{proof}

\begin{proposition}\label{prop:RightshiftCPTP}
Family $\{\Delta_\tau : \tau \geqslant 0\}$ of right shift operators is a contraction semigroup of completely positive and trace preserving maps on $\mathscr{L}^2 (\mathbb{T},\matrd)$. Likewise, $\{\tilde{\Delta}_\tau = I \otimes e^{-i\tau\Omega F_z} : \tau \geqslant 0\}$ is a contraction semigroup of completely positive and trace preserving maps on $\matrd_1 \bar{\otimes} L^2(\mathbb{T})$.
\end{proposition}

\begin{proof}
Family $\{\Delta_\tau\}$ is easily shown to be a semigroup of CP maps. As the Lebesgue measure is traslationally invariant, for any periodic measurable $\varphi$ and any $\delta\in\reals$ we have $\int_{[0,T]} \varphi \, d\mu = \int_{[\delta,\delta+T]}\varphi \, d\mu$ yielding
\begin{equation}
	\trl{\Delta_\tau (f)} = \frac{1}{T}\int_{0}^{T} \tr{f(t-\tau)} \, dt = \frac{1}{T}\int_{-\tau}^{T-\tau} \tr{f(t)}\, dt = \trl{f}
\end{equation}
for any $f\in\mathscr{L}^2(\mathbb{T},\matrd_1)$, so $\Delta_\tau \in \tp{\mathscr{L}^2}$; as the above computation remains true after changing $f(t-\tau)$ for $\tnorm{f(t-\tau)}^{2}$ under the integral, the equality
\begin{equation}
	\| \Delta_\tau (f) \|_{\mathscr{L}^{2}_{1}} = \| f \|_{\mathscr{L}^{2}_{1}}
\end{equation}
emerges, so it is a contraction semigroup. Claims related to family $\{\tilde{\Delta}_\tau\}$ then follow.
\end{proof}

\subsection{Evolution in generalized space of states}
\label{subsect:MEinSspace}

\subsubsection{Generalized Lindbladian}

Now we are ready to apply the general construction outlined in previous sections to define the dynamical map in the generalized space of states. First, let us set operators $L, P$ on $\mathscr{L}^2(\mathbb{T},\matrd_1)$ given as in section \ref{sect:Operatorvaluedfunction} as operator-valued functions
\begin{equation}
	f\mapsto L(f)(t) = L_t (f(t)), \qquad f\mapsto P(f)(t) = P_t (f(t)),
\end{equation}
where $t\mapsto P_t$ and $t\mapsto L_t$ were originally given in section \ref{sect:ExamplesofperiodicLindbladians} by formulas
\begin{equation}
	L_t (a) = -i \comm{H_t}{a} + P_t K P_{t}^{-1} (a), \quad P_t (a) = p_t \, a \, p_{t}^{\hadj}.
\end{equation}
Then, from cyclic properties of trace, one can easily obtain 
\begin{equation}
	(a, P_t (b))_{\matrd} = \tr{a p_t b p_{t}^{\hadj}} = \tr{p_{t}^{\hadj} a p_t b} = (P_{t}^{\dual} (a), b)_{\matrd},
\end{equation}
so $P_{t}^{\dual}(a) = p_{t}^{\hadj} a p_t$. Function $p : \mathbb{T} \to p_t \in \matrd$ is bounded and therefore square integrable; then, there exists $\tilde{p} \in \matrd_1 \bar{\otimes} L^2(\mathbb{T})$ such that
\begin{equation}
	\tilde{p} = \sum_{n\in\integers} p_n \otimes e_n , \quad p = \sum_{n\in\integers} p_n \odot e_n, \quad p_n = \frac{1}{T}\int_{\mathbb{T}} p_t e^{-in\Omega t} \, dt,
\end{equation}
and $p_t \stackrel{\mathrm{a.e.}}{=} \sum_{n\in\integers} p_n e^{in\Omega t}$. Then, for $\tilde{p}^{\adj} = \sum_{n\in\integers} p_{n}^{\hadj} \otimes e_{n}^{\adj}$ we have $p^\adj p = I \odot e_0$ and $\tilde{p}^\adj \tilde{p} = I \otimes e_0$, so $p$, $\tilde{p}$ are unitary. From this we have
\begin{equation}
	P (f) = p \, f \, p^\adj, \quad P^\dual (f) = p^\adj \, f \, p, \quad \tilde{P} (\tilde{f}) = \tilde{p} \, \tilde{f} \, \tilde{p}^\adj, \quad \tilde{P}^{\dual} (\tilde{f}) = \tilde{p}^\adj \, \tilde{f} \, \tilde{p}.
\end{equation}
Applying \eqref{eq:UnitaryEvolutionOperator} it is easy to obtain
\begin{equation}
	\frac{dp_t}{dt} = -i H_t p_{t} + i p_{t} \bar{H}, \quad \frac{dp_{t}^{\hadj}}{dt} = \Big(\frac{dp_t}{dt}\Big)^{\hadj},
\end{equation}
as well as, after some algebra, derivatives $\dot{P}$, $\dot{P}^\dual$,
\begin{equation}\label{eq:PdotPdotDual}
	\dot{P}_t = -i \comm{H_t-P_t(\bar{H})}{P_t}, \qquad \dot{P}^{\dual}_{t} = i \comm{P^{\dual}_{t}(H_t)-\bar{H}}{P^{\dual}_{t}}.
\end{equation}

\begin{proposition}\label{prop:PfunctionConvergence}
The following hold:
\begin{enumerate}
	\item\label{point:PfunctionConvergencePFourier} Fourier series $\sum_{n\in\integers} P_n \odot e_n$ converges absolutely and uniformly to $P$,
	\item\label{point:PfunctionConvergenceDerFourier} $\dot{P}$ and $\dot{P}^{\dual}$ admit uniformly convergent Fourier series
	\begin{equation}\label{eq:PfunctionConvergenceDerFourierSeries}
		\dot{P} = i\Omega \sum_{n\in\integers} n P_n \odot e_n, \qquad \dot{P}^{\dual} = i\Omega \sum_{n\in\integers} n P_{n}^{\dual} \odot e_n.
	\end{equation}
	\item\label{point:PpropertiesUnitary} $P^{\dual}P = PP^{\dual} = I \odot e_0$ and $\tilde{P}^\dual \tilde{P} = \tilde{P} \tilde{P}^\dual = I \otimes e_0$,
	\item\label{point:PpropertiesDistributive} $P$ and $\tilde{P}$ are distributive and preserve adjoints,
	\item\label{point:PpropertiesCPcontractice} $P \in \cptp{\mathscr{L}^2}$, $\tilde{P}\in \cptp{\mathscr{\bar{\otimes}}}$ and both $P$, $\tilde{P}$ are isometries.
\end{enumerate}
\end{proposition}

\begin{proof}
Ad (\ref{point:PfunctionConvergencePFourier}). As $\dot{P}$ is bounded, this claim is a consequence of Lemma \ref{lemma:MatrixFunFourierAbsConv} (see Section \ref{appsection:FourierSeries} in the Appendix).

Ad (\ref{point:PfunctionConvergenceDerFourier}). Second derivatives of $P$ and $P^\dual$ are easy to compute (we will omit the explicit calculations) and one immediately notices that $\ddot{P}$ and $\ddot{P}^{\dual}$ exist whenever $dH_t / dt$ exists. From assumption of a.e. differentiability of $H_t$ we conclude $\ddot{P}$ and $\ddot{P}^\dual$ are piecewise continuous; therefore by Lemma \ref{lemma:MatrixFunFourierUniConv}, both $\dot{P}$ and $\dot{P}^\dual$ admit uniformly convergent Fourier series.

Ad (\ref{point:PpropertiesUnitary}). This comes easily from unitarity of $\tilde{p}$ and $p$.

Ad (\ref{point:PpropertiesDistributive}). Again, employing unitarity of $p$,
\begin{equation}
	P (fg) = p \, fg \, p^{\adj} = p \, f \, p^{\adj} p \, g \, p^{\adj} = P(f)P(g)
\end{equation}
for any $fg \in \mathscr{L}^2 (\mathbb{T},\matrd)$, and
\begin{equation}
	P(f^\adj) = p \, f^\adj \, p^\adj = (p \, f \, p^\adj)^\adj = P(f)^\adj.
\end{equation}
Similar claims remain true for Fourier liftings.

Ad (\ref{point:PpropertiesCPcontractice}). From cyclicity of trace and unitarity of $p_t$ we have $\tr{P_t (a)} = \tr{a}$, so $P_t \in \tp{\matrd}$ for each $t\in\mathbb{T}$ and $P,P^\dual\in\tp{\mathscr{L}^2}$ by Proposition \ref{thm:AlphaTP}. As $P_t (a) = p_t a p_{t}^{\hadj}$ is clearly of Kraus form for each $t\in\mathbb{T}$ and $t\mapsto p_t$ is a bounded function, Proposition \ref{thm:AlphaAtCP} yields $P, P^\dual \in \cp{\mathscr{L}^2}$ and in consequence $\tilde{P},\tilde{P}^\dual \in \cp{\mathscr{\bar{\otimes}}}$. 
For any $a\in\matrd$ we have $\tnorm{a}=\tnorm{a^\hadj}$ \cite{Conway1999}; this yields that for any unitary matrices $u,v$, mapping $a \mapsto u a v$ is an isometry,
\begin{align}
	\tnorm{uav} &= \tr{\sqrt{v^\hadj a^\hadj u^\hadj u a v}} = \tr{\sqrt{(av)^{\hadj}av}} = \tnorm{av} \\
	&= \tnorm{v^\hadj a^\hadj} = \tr{\sqrt{av v^\hadj a^\adj}} = \tr{\sqrt{aa^\hadj}} = \tnorm{a^\hadj} = \tnorm{a},\nonumber
\end{align}
therefore for $a \mapsto P_t(a) = p_t a p_{t}^{\hadj}$,
\begin{equation}
	\| P_t \|_{\infty} = \sup_{\tnorm{a}\leqslant 1}\tnorm{p_t a p_{t}^{\hadj}} = \sup_{\tnorm{a}\leqslant 1}\tnorm{a} = 1
\end{equation}
since $p_t$ is unitary and $P_t$ is an isometry. Switching $p_t$ with $p_{t}^{\hadj}$ also yields $P_{t}^{-1}$ is an isometry as well. This yields
\begin{equation}
	\| P(f) \|_{\mathscr{L}^{2}_{1}}^{2} = \frac{1}{T} \int_{\mathbb{T}} \tnorm{P_t (f(t))}^{2}\, dt = \frac{1}{T} \int_{\mathbb{T}} \tnorm{f(t)}^{2}\, dt = \| f \|_{\mathscr{L}^{2}_{1}}^{2},
\end{equation}
which shows $P$, $\tilde{P}$ are isometries.
\end{proof}

\begin{lemma}\label{prop:LtSquareInt}
$L \in \mathscr{L}^{\infty}(\mathbb{T},B(\matrd))$ and therefore it is square-integrable.
\end{lemma}
\begin{proof}
Pick any $t\in\mathbb{T}$; then
\begin{align}
	\| L_t \|_{\infty} \leqslant \| \comm{H_t}{\cdot} \|_{\infty} + \| P_t K P_{t}^{-1} \|_{\infty}.
\end{align}
Trivially, for any two matrices $a,b\in \matrd$, one has $\tnorm{\comm{a}{b}} \leqslant 2 \tnorm{a} \tnorm{b}$ (we note, however, that one can provide a better upper bound; see \cite{Boettcher2008} for details). As any two norms on $\matrd$ are equivalent \cite{RogerA.Horn2012}, there exists a constant $\frac{1}{2} C>0$ such that $\tnorm{\cdot} \leqslant \frac{1}{2} C \| \cdot \|_{\infty}$. This yields
\begin{subequations}
	\begin{equation}
		\| \comm{H_t}{\cdot} \|_{\infty} = \sup_{\tnorm{a}\leqslant 1}{\tnorm{H_t a - a H_t}} \leqslant C \|H_t\|_{\infty},
	\end{equation}
	\begin{align}\label{eq:LtSquareIntSecondEq}
		\| P_t K P_{t}^{-1} \|_{\infty} &= \sup_{\tnorm{a}\leqslant 1} \tnorm{P_t KP_{t}^{-1}(a)} = \sup_{\tnorm{a}\leqslant 1} \tnorm{KP_{t}^{-1}(a)} \\
		&\leqslant \| K \|_{\infty} \sup_{\tnorm{a}\leqslant 1} \tnorm{P_{t}^{-1}(a)} = \| K \|_{\infty}.\nonumber
	\end{align}
\end{subequations}
where \eqref{eq:LtSquareIntSecondEq} comes from Proposition \ref{prop:PfunctionConvergence}. So we have
\begin{equation}
	\| L \|_{\mathscr{L}^\infty} = \sup_{t\in\mathbb{T}}\| L_t \|_{\infty} \leqslant C D + \|K\|_{\infty} < \infty
\end{equation}
where $D = \sup_{t\in\mathbb{T}}\|H_t\|_{\infty}$. Hence, $L$ is bounded. Square integrability is then immediate, as $\int_{\mathbb{T}} \| L_t \|_{\infty}^{2} \, dt \leqslant \| L \|_{\mathscr{L}^\infty}^{2}$.
\end{proof}

The dynamics in the enlarged space of states is then defined by a new object, which we call a \emph{generalized Lindbladian}; this should be understood as an extension of the original Lindbladian $L_t$, encoding information about periodicity of fundamental solutions; we remark that the following definition of generalized Lindbladian corresponds exactly to the \emph{generalized Hamiltonian} (also called the \emph{Floquet Hamiltonian} in some sources), constructed in similar manner in Howland time-independent formalism. Unsurprisingly, these two closely related objects share some properties, as for example a band-like structure of a spectrum, as we show.

\begin{definition}\label{def:Lmap}
Densely defined map $\mathcal{L} : \mathscr{C}^1 (\mathbb{T},\matrd_1) \to \mathscr{L}^2 (\mathbb{T},\matrd_1)$ given as
\begin{equation}\label{eq:Lmap}
	\mathcal{L} = L - \partial, \quad f \mapsto \mathcal{L}(f)(t) = \Big( L_t - \frac{d}{dt} \Big) (f(t))
\end{equation}
will be called the generalized Lindbladian. Its Fourier lifting $\tilde{\mathcal{L}}$, densely defined on $\iota^{-1}(\mathscr{C}^{1}(\mathbb{T},\matrd_1))$ will be then simultaneously given by
\begin{equation}\label{eq:LmapS}
	\tilde{\mathcal{L}} = \tilde{L} - i\Omega \, I \otimes F_z,
\end{equation}
and often referred to by the same name.
\end{definition}

Recall, that Floquet theorem allowed to postulate the Floquet normal form of the solution to the MME \ref{eq:LtFloquet} given by $\Lambda_t = P_t e^{t\bar{L}}$ for $\bar{L} = -i \ad{\bar{H}} + K$ \eqref{eq:LambdaGeneral}. Assume $\bar{L}$ to be diagonalizable by family of linearly independent matrices $\{ \varphi_j \}\subset\matrd_1$,
\begin{equation}
	\bar{L}(\varphi_j) = \xi_j \varphi_j, \qquad \xi_j \in \complexes, \, 1\leqslant j \leqslant d^2.
\end{equation}

\begin{proposition}
Point spectrum of $\mathcal{L}$ is of a form 
\begin{equation}
	\sigma_{p}(\mathcal{L}) = \sigma (\bar{L}) - i \Omega \integers, \quad n\in\integers
\end{equation}
and the corresponding eigenbasis is
\begin{equation}
	\{ \phi_{j,n} = P(\varphi_j \odot e_n) : \bar{L}(\varphi_j) = \xi_j \varphi_j , \, n\in\integers\}.
\end{equation}
Similarily, family $\{ \tilde{P}(\varphi_j \otimes e_n) \}$ is the eigenbasis of $\tilde{\mathcal{L}}$ for the same eigenvalues.
\end{proposition}

\begin{proof}
It suffices to show $\phi_{j,n}$ satisfy the eigenequation
\begin{equation}\label{eq:LtMinusDerEigenequation}
	L_t (\phi_{j,n}(t)) - \frac{d}{dt} \phi_{j,n}(t) = \xi_{j,n} \phi_{j,n}(t).
\end{equation}
Note, that $\phi_{j,n}(t) = e^{-t\xi_{j,n}} \varphi_j (t)$, where $\varphi_j (t) = e^{t\xi_j} \phi_j (t)$; substituting this back to \eqref{eq:LtMinusDerEigenequation} and employing the fact, that $\varphi_j (t)$ solves the MME \eqref{eq:LtFloquet}, i.e.~$L_t(\varphi_j(t)) = \dot{\varphi}_j (t)$, the result is immediate.
\end{proof}

\subsubsection{CP-divisible dynamics in generalized space}

In this section we will introduce a dynamics on the enlarged space, induced by generalized Lindbladian, as mentioned earlier; we partially follow the reasoning given originally by Howland \cite{Howland1974,Howland1979,Cycon2007}. Let an operator-valued function $\tau \mapsto \mathcal{V}_\tau \in B(\mathscr{L}^2)$, $\tau \in [0,\infty )$, be defined via
\begin{equation}
	\mathcal{V}_\tau (f)(t) = V_{t,t-\tau}(f(t)),
\end{equation}
where $V_{t,s} = \Lambda_t \Lambda_{s}^{-1}$ is the propagator of quantum dynamical map $\Lambda_t$. We will consider a composition of $\mathcal{V}_\tau$ with a right shift operator $\Delta_\tau$ given by \eqref{eq:RightShiftDefinition}, namely a function $\tau\mapsto W_\tau = \mathcal{V}_\tau \circ \Delta_\tau$, acting on $\mathscr{L}^2 (\mathbb{T},\matrd_1)$ via
\begin{equation}\label{eq:WtauDef}
	W_\tau(f)(t) = V_{t,t-\tau}(f(t-\tau)),
\end{equation}
as well as its Fourier lifting $\tilde{W}_\tau$. The important result then follows:

\begin{theorem}\label{thm:WtauCPTPcontraction}
The following claims hold:
\begin{enumerate}
	\item\label{claim:WtauCPTPcontractionOne} Families $\{W_\tau : \tau \geqslant 0\}$ and $\{\tilde{W}_\tau : \tau \geqslant 0\}$ are strongly differentiable contraction $C_0$-semigroups of completely positive and trace preserving maps over $\mathscr{L}^{2}(\mathbb{T},\matrd_1)$ and $\complexes^{n\times n} \bar{\otimes} L^2 (\mathbb{T})$, generated by $\mathcal{L}$ and $\tilde{\mathcal{L}}$, respectively, i.e.~$W_\tau = e^{\tau\mathcal{L}}$ and $\tilde{W}_\tau = e^{\tau\tilde{\mathcal{L}}}$.
	\item\label{claim:WtauCPTPcontractionTwo} Semigroup $\{e^{\tau \tilde{\mathcal{L}}} : \tau \geqslant 0\}$ admits factorized form
		\begin{equation}
			e^{\tau\tilde{\mathcal{L}}} = \tilde{P} (e^{\tau \bar{L}} \otimes e^{-i\tau\Omega F_z}) \tilde{P}^{\dual}.
		\end{equation}
\end{enumerate}
\end{theorem}

\begin{proof}
Ad \eqref{claim:WtauCPTPcontractionOne}. We prove the statement only for $W_\tau$, as the proof for $\tilde{W}_\tau$ can be performed in analogous manner. Applying general ideas by Howland \cite{Howland1974,Cycon2007} we first check the semigroup properties. Clearly, $W_\tau = I$; applying Chapman-Kolmogorov properties of $V_{t,s}$ for $\tau_1, \tau_2 \in [0,\infty )$ (e.g. divisibility) we obtain
\begin{align}
	W_{\tau_1 + \tau_2}(f)(t) &= V_{t,t-\tau_1 - \tau_2}(f(t-\tau_1-\tau_2)) \\
	&= V_{t,t-\tau_1} V_{t-\tau_1, t-\tau_1-\tau_2} (f(t-\tau_1-\tau_2)) \nonumber \\
	&= V_{t,t-\tau_1}( W_{\tau_2}(f)(t-\tau_1)) = W_{\tau_1} W_{\tau_2} (f)(t) \nonumber
\end{align}
and so $W_{\tau_1+\tau_2} = W_{\tau_1}W_{\tau_2}$.

As $\Lambda_t$ is CP-divisible, $V_{t,t-\tau}\in \cptp{\matrd}$ and is a trace norm contraction for every $t \in [0,\infty )$. Propositions \ref{thm:AlphaTP} and \ref{prop:AlphaContractive} yield $\mathcal{V}_\tau \in \tp{\mathscr{L}^2}$ and $\mathcal{V}_\tau$ is a contraction. For $f(t) = a$, the constant function, \eqref{eq:LambdaGeneral} yields, after simple algebra
\begin{equation}
	\mathcal{V}_{\tau}(f)(t) = V_{t,t-\tau}(a) = p_t e^{\tau\bar{L}}(p_{t-\tau}^{\hadj} a\, p_{t-\tau})p_{t}^{\hadj}.
\end{equation}
As $e^{\tau\bar{L}} \in \cptp{\matrd}$, it admits Kraus representation $e^{\tau\bar{L}} = \sum_{j=1}^{d^2} Y_{j,\tau}^{\hadj} a Y_{j,\tau}$. This allows to write $V_{t,t-\tau}$ in Kraus form as well,
\begin{equation}
	V_{t,t-\tau}(a) = \sum_{j=1}^{d^2} X_{j\tau, t}^{\hadj} a X_{j\tau,t}, \qquad X_{j\tau,t} = p_{t-\tau}Y_{j,\tau}p_{t}^{\hadj},
\end{equation}
where functions $t\mapsto X_{j\tau,t}$ are clearly bounded. Then, Proposition \ref{thm:AlphaAtCP} guarantees $\mathcal{V}_{\tau} \in \cp{\mathscr{L}^2}$. As $\Delta_\tau \in \cptp{\mathscr{L}^2}$ and is a contraction (Proposition \ref{prop:RightshiftCPTP}), we finally have $W_\tau \in \cptp{\mathscr{L}^2}$ and $W_\tau$ is a contraction as well. Claims related to $\tilde{W}_\tau$ follow simultaneously.

To show that both families are generated by claimed maps, it suffices to compute the strong derivative of $\tau \mapsto W_\tau$. We have, for any $f\in \mathscr{C}^1 (\mathbb{T},\matrd_1)$,
\begin{align}
	\lim_{\tau\searrow 0}&\frac{1}{\tau}(W_\tau - I)(f)(t) = \lim_{\tau\searrow 0} \frac{1}{\tau} (V_{t,t-\tau}(f(t-\tau))-f(t)) \\
	&= \Lambda_t \lim_{\tau\searrow 0} \frac{1}{\tau} (\Lambda_{t-\tau}^{-1}(f(t-\tau))-\Lambda_{t}^{-1}(f(t))) = \Lambda_t \left. \frac{\partial g(t,\tau)}{\partial \tau}\right|_{\tau = 0}, \nonumber
\end{align}
where $g(t,\tau) = \Lambda_{t-\tau}^{-1}(f(t-\tau))$. The derivative of $g(t,\cdot)$ is easily found to be
\begin{equation}
	\left. \frac{\partial g(t,\tau)}{\partial \tau}\right|_{\tau = 0} = -\dot{\Lambda}_{t}^{-1}(f(t)) - \Lambda_{t}^{-1}( \dot{f}(t))
\end{equation}
leading to
\begin{equation}
	\lim_{\tau\searrow 0}\frac{1}{\tau}(W_\tau - I)(f)(t) = L_t(f(t)) - \frac{df(t)}{dt} = \mathcal{L}(f)(t),
\end{equation}
where the limit exists for every $f\in \mathscr{C}^1 (\mathbb{T},\matrd_1)$ and $W_\tau = e^{\tau\mathcal{L}}$; then, the remaining claim $\tilde{W}_\tau = e^{\tau \tilde{\mathcal{L}}}$ is obtained after considering appropriate Fourier lifting.

Ad \eqref{claim:WtauCPTPcontractionTwo}. Notice, that one can write
\begin{equation}
	\tilde{P} (e^{\tau \bar{L}} \otimes e^{-i\tau\Omega F_z}) \tilde{P}^{\dual} = \tilde{P}(e^{\tau \bar{L}}\otimes I)\tilde{\Delta}_{\tau}\tilde{P}^{\dual},
\end{equation}
where $\tilde{\Delta}_\tau = I \otimes e^{-i\tau\Omega F_z}$ is the Fourier lifting of right shift operator, as given in Proposition \ref{prop:DeltaTauSrep}; this yields
\begin{align}
	\iota \circ \tilde{P} (e^{\tau \bar{L}} &\otimes e^{-i\tau\Omega F_z}) \tilde{P}^{\dual}(\tilde{f})(t) = P_t e^{\tau \bar{L}} P_{t-\tau}^{\dual}(f(t-\tau)) \\
	&= P_t e^{t\bar{L}} e^{-(t-\tau) \bar{L}} P_{t-\tau}^{\dual} (f(t-\tau)) = \Lambda_t \Lambda_{t-\tau}^{-1} (f(t-\tau)) \nonumber \\
	&= V_{t,t-\tau} (f(t-\tau)) = e^{\tau \mathcal{L}} (f)(t),\nonumber
\end{align}
which means $\tilde{P} (e^{\tau \bar{L}} \otimes e^{-i\tau\Omega F_z}) \tilde{P}^{\dual} = \iota^{-1} \circ e^{\tau \mathcal{L}} \circ \iota = e^{\tau\tilde{\mathcal{L}}}$, as claimed. 
\end{proof}

The next result is a straightforward, yet important practical implication of preceding Theorem.

\begin{proposition}
Solution of the MME \eqref{eq:LtFloquet} may be expressed as
\begin{equation}
	\rho_t = e^{t\mathcal{L}}(\rho_0 \odot e_0)(t) = \iota \circ e^{t\tilde{\mathcal{L}}}(\rho_0 \otimes e_0) (t),
\end{equation}
where $\rho_0 \in \matrd$ is the initial density operator, $\rho \geqslant 0$, $\tr{\rho_0} = 1$.
\end{proposition}

\begin{proof}
Set $\rho_0 \odot e_0 \in \mathscr{L}^{2}(\mathbb{T},\matrd_1)$. By straightforward consequence of \eqref{eq:WtauDef} and Theorem \ref{thm:WtauCPTPcontraction},
\begin{equation}
	e^{t\mathcal{L}}(\rho_0 \odot e_0)(t) = V_{t,t-t}((\rho_0 \odot e_0)(0)) = V_{t,0}(\rho_0)
\end{equation}
which is equal to $\Lambda_t (\rho_0)$, as $V_{t,s} = \Lambda_t \Lambda_{s}^{-1}$ and $\Lambda_0 = I$. The second equality follows after putting $\rho_0 \otimes e_0 = \iota^{-1}(\rho_0 \odot e_0)$ and $e^{\tau\mathcal{L}} = \iota \circ e^{\tau\tilde{\mathcal{L}}} \circ \iota^{-1}$.
\end{proof}

Implications of the above Proposition are severe: namely, as a density operator, i.e.~a solution of original Master Equation, can be effectively computed by employing map $e^{t\mathcal{L}}$, such map may indeed be considered a \emph{generalized evolution} of the system. Theorem \ref{thm:WtauCPTPcontraction} then shows that such evolution inherits all nice mathematical properties from its counterpart acting on space $\matrd_1$, i.e.~is also completely positive, trace preserving and a contraction. As we emphasized few times before, this corresponds nicely with original Howland formulation, where the generalized evolution, induced by Floquet Hamiltonian, was given by \emph{unitary group} acting on the enlarged (Hilbert) space.

\subsection{Fourier formulation of time-independent formalism}
\label{sec:FourierExpansions}

In usual approach to time-independent formalism, say in NMR analysis \cite{Ernst2005,Scholz2010}, one often finds informative and useful to work with explicit, Fourier-like expansions of certain Fourier liftings; and so, given a time-periodic operator $A_t = A_{t+T}$ (for example the NMR Hamiltonian), one often expresses its related Fourier lifting $\tilde{A}$ as a series
\begin{equation}\label{eq:FourierLikeExpr}
	\tilde{A} \sim \sum_{n\in\integers} A_n \otimes F_n,
\end{equation}
where $A_n = \frac{1}{T} \int_{\mathbb{T}} A_t e^{-in\Omega t} \, dt$ stands for the Fourier transform of $A_t$ and $\{F_n\}$ are unitary Fourier shift operators, as given by Definition \ref{def:FourierShifts}. In this section, we explore such Fourier-like expansions of various time-dependent operators in some more detail and address some convergence-related issues.

\subsubsection{Fourier-like expression for \texorpdfstring{$\tilde{A}$}{S-representation of A} and its convergence}

Here we examine convergence of Fourier-like expressions \eqref{eq:FourierLikeExpr} in case of some particular operator $A$ on $\mathscr{L}^{2}(\mathbb{T},\matrd_1)$ defined as in Section \ref{sect:Operatorvaluedfunction} by $A(f)(t) = A_t (f(t))$, where $t\mapsto A_t$ is an operator-valued function. We explicitly note this function as $\mathbf{A} : \mathbb{T}\to B(\matrd_1)$, $\mathbf{A}(t) = A_t$. In all the following, we put
\begin{equation}
	A_n = \frac{1}{T} \int_{\mathbb{T}} A_t e^{-in\Omega t} \, dt.
\end{equation}

\begin{proposition}\label{prop:NormConvergenceJA}
If Fourier series $\sum_{n\in\integers} A_n \odot e_n$ converges uniformly to $\mathbf{A}$, then $\sum_{n\in\integers} A_n \otimes F_n$ converges to $\tilde{A}\in B(\matrd_1 \, \bar{\otimes} \, L^2 (\mathbb{T}))$ in norm. Likewise, $\sum_{n\in\integers} A_{n}^{\dual} \otimes F_{n}$ converges to $\tilde{A}^{\dual}$ in norm.
\end{proposition}

\begin{proof}
It is easy to see, that
\begin{align}\label{eq:NormConvergenceJA}
	\Big\| &\sum_{|k|\leqslant n} A_{k} \otimes F_k - \tilde{A} \Big\|_{\infty}^{2} = \sup_{\|f\|\leqslant 1} \frac{1}{T}\int_{\mathbb{T}} \Big\|\big(D_n * \mathbf{A}(t)-A_t\big)(f(t))\Big\|_{1}^{2} \, dt \\
	&\leqslant \sup_{t\in\mathbb{T}} \big\| D_n * \mathbf{A}(t) -A_t \big\|_{\infty}^{2} \sup_{\|f\|\leqslant 1} \frac{1}{T}\int_{\mathbb{T}} \tnorm{f(t)}^{2} \, dt = \big\| D_n * \mathbf{A} - \mathbf{A} \big\|_{\mathscr{L}^{\infty}}^{2}, \nonumber
\end{align}
which tends to $0$ as $n\to\infty$, as Fourier series $\sum_{n\in\integers} A_n \odot e_n$ converges uniformly to $\mathbf{A}$. Then, Fourier series $\sum_{n\in\integers}A_{n}^{\dual} \odot e_n$ converges uniformly to $A^\dual$ and \eqref{eq:NormConvergenceJA} may be directly reapplied to show that $\sum_{n\in\integers} A_{n}^{\dual} \otimes F_n$ converges to $\tilde{A}^{\dual}$.
\end{proof}

\begin{proposition}\label{prop:SreprAconvergence}
Let $\mathbf{A} : \mathbb{T} \to B(\matrd_1)$ be bounded. Then, $\tilde{A}$ is bounded and series $\sum_{n\in\integers} A_n\otimes F_n$ and $\sum_{n\in\integers} A_{n}^{\dual}\otimes F_n$ converge to $\tilde{A}$ and $\tilde{A}^{\dual}$, respectively, pointwise on $\iota^{-1}(\mathscr{C}^{0}(\mathbb{T},\matrd_1))$.
\end{proposition}

\begin{proof}
Boundedness of $A$ and $A^\dual$ implies $A, \, A^\dual \in \mathscr{L}^2 (\mathbb{T},B(\matrd_1))$, so by Theorem \ref{FourierSeriesConvLp}, their Fourier series converge in $\| \cdot \|_{\mathscr{L}^{2}_{1}}$ norm. Take any $f\in \mathscr{C}^{0}(\mathbb{T},\matrd_1)$; employing boundedness of $A_{k}$ and isometry properties of $\iota$ we have
\begin{align}
	\Big\| \sum_{|k|\leqslant n}A_{k} \otimes F_k (\tilde{f}) - \tilde{A}(\tilde{f}) \Big\|_{\mathscr{L}^{2}_{1}}^{2} &= \frac{1}{T}\int_{\mathbb{T}} \Big\|\big(D_n * \mathbf{A}(t) - A_t\big)(f(t))\Big\|^{2}_{1} \, dt \\
	&\leqslant \sup_{t\in\mathbb{T}} \tnorm{f(t)}^{2} \cdot  \frac{1}{T}\int_{\mathbb{T}} \big\| D_n * \mathbf{A} (t)-A_t \big\|_{\infty}^{2} dt,\nonumber
\end{align}
which tends to $0$, $n\to\infty$, as $f$ attains its maximum over $\mathbb{T}$. The second equality results analogously.
\end{proof}

\begin{proposition}
If Fourier series of function $\mathbf{A} : \mathbb{T} \to B(\matrd_1)$ converges absolutely, then $\sum_{n\in\integers} A_n \otimes F_n$ converges strongly to $\tilde{A}$. Likewise, $\sum_{n\in\integers} A_{n}^{\dual} \otimes F_n$ converges strongly to $\tilde{A}^\dual$.
\end{proposition}

\begin{proof}
Sequence $(T_n)\subset \mathcal{X}$ of bounded operators on Banach space $\mathcal{X}$ converges in strong operator topology to some $T \in B(\mathcal{X})$, if and only if \cite{Eidelman2004}
\begin{enumerate}
	\item for all $x \in M$, where $M$ is dense in $\mathcal{X}$, we have $(T_n - T)(x) \to 0$, and
	\item $(\| T_n \|_\infty)\in l^\infty$.
\end{enumerate}
Take
\begin{equation}
	T_n = \sum_{|k|\leqslant n} A_k \otimes F_k .
\end{equation}
As $\iota^{-1}(\mathscr{C}^{0}(\mathbb{T},\matrd_1))$ is a dense subspace of $\matrd_1 \, \bar{\otimes} \, L^2 (\mathbb{T})$, first condition is automatically fulfilled by Proposition \ref{prop:SreprAconvergence}. Operator norm in $B(\matrd_1 \, \bar{\otimes} \, L^2 (\mathbb{T}))$ is a cross-norm and $\|F_n\|_\infty =1$, so we have
\begin{equation}
	\sup_{n\in\integers} \| T_n \|_\infty \leqslant \sup_{n\in\integers} \sum_{|k|\leqslant n} \| A_k \|_\infty = \sum_{n\in\integers} \| A_n \|_\infty
\end{equation}
which is finite from absolute convergence of Fourier series. Thus, $(\| T_n \|)$ is bounded and strong convergence is shown (proof for the adjoint series is analogous).
\end{proof}

\begin{lemma}\label{lemma:AtDiffEndomorphism}
If function $\mathbf{A} : \mathbb{T} \to B(\matrd_1)$ is continuous and of bounded derivative, then $A$ is an endomorphism over $\mathscr{C}^{1}(\mathbb{T},\matrd_1)$.
\end{lemma}

\begin{proof}
As linearity is obvious, we need to show $A (\mathscr{C}^1) \subset \mathscr{C}^1$. Take any function $f\in \mathscr{C}^1 (\mathbb{T},\matrd_1)$; it suffices to find such continuous $\xi$, that for any $t\in\mathbb{T}$
\begin{equation}\label{eq:AtDerivative}
	\lim_{h\to 0} \tnorm{h^{-1} \left[ A_{t+h}(f(t+h)) - A_t(f(t)) \right] - \xi(t)} = 0.
\end{equation}
In fact one can easily show, that
\begin{equation}\label{eq:XiFunDefinition}
	\xi(t) = \dot{A}_{t}(f(t)) + A_t (\dot{f}(t))
\end{equation}
by adding and subtracting $\frac{1}{h} A_t (f(t+h))$ under the norm in \eqref{eq:AtDerivative}, applying \eqref{eq:XiFunDefinition} and reordering terms; we then obtain the upper bound on l.h.s. of \eqref{eq:AtDerivative},
\begin{align}\label{eq:AtDiffEndomorphismUpperBound}
	&\tnorm{h^{-1}[A_{t+h}-A_t] (f(t) + O(h))} \\
	+ &\tnorm{A_t \left( h^{-1}[f(t+h)-f(t)] - f^{\prime}(t) \right)} \nonumber \\
	\leqslant &\big\| h^{-1}[A_{t+h}-A_t] \big\|_{\infty} \tnorm{O(h)} + \big\| h^{-1}[A_{t+h}-A_t] - \dot{A}_t \big\|_{\infty} \tnorm{f(t)} \nonumber
\end{align}
where $O(h) = f(t+h) - f(t)$. Due to continuity of $f$ and boundedness of $A^{\prime}$, \eqref{eq:AtDiffEndomorphismUpperBound} tends to 0 as $h\to 0$.
\end{proof}

\begin{proposition}\label{prop:NormConvergenceProductofFunctions}
Let $\mathbf{A},\mathbf{B} : \mathbb{T} \to B(\matrd_1)$ be bounded functions of uniformly convergent Fourier series $\sum_{n\in\integers}A_n \odot e_n$, $\sum_{n\in\integers}B_n \odot e_n$, respectively, and let one of the series be additionally absolutely convergent. Then, $\tilde{A}\tilde{B}$ can be expressed as norm-convergent series
\begin{equation}
	\tilde{A}\tilde{B} = \sum_{n,m\in\integers} A_{n}B_{m} \otimes F_{n}F_{m}.
\end{equation}
\end{proposition}

\begin{proof}
Without loss of generality assume $\sum_{n\in\integers}A_{n} \odot e_n$ converges absolutely. Then, again employing isometry properties of $\iota$, we have
\begin{align}\label{eq:NormConvergenceProductofFunctionsIneq}
	&\Big\| \sum_{|k|\leqslant n}\sum_{|l|\leqslant m} A_{k}B_{l} \otimes F_{k}F_{l} - \tilde{A}\tilde{B} \Big\|_\infty \\
	&\leqslant \sup_{\|\tilde{f}\|\leqslant 1} \Bigg( \frac{1}{T}\int_{\mathbb{T}} \Big\| \sum_{|k|\leqslant n}\sum_{|l|\leqslant m} A_{k}B_{l} e^{i(k+l)\Omega t} - A_t B_t \Big\|^{2}_{\infty} \tnorm{f(t)}^{2} \, dt \Bigg)^{1/2} \nonumber \\
	&\leqslant \sup_{t\in\mathbb{T}} \Big\| \sum_{|k|\leqslant n}\sum_{|l|\leqslant m} A_{k}B_{l} e^{i(k+l)\Omega t} - A_t B_t \Big\|_{\infty}. \nonumber
\end{align}
Adding and subtracting $\sum_{|k|\leqslant n} A_{k} e^{ik\Omega t} B_t$ under the norm and employing triangle inequality one can find the upper bound of \eqref{eq:NormConvergenceProductofFunctionsIneq} to be
\begin{align}
	\sum_{|k|\leqslant n} \| A_{k} \|_{\infty} \sup_{t\in\mathbb{T}} \| &D_m * \mathbf{B} (t) - B_t \|_{\infty} + \sup_{t\in\mathbb{T}} \| B_t \|_{\infty} \| D_n * \mathbf{A}(t) - A_t\|_{\infty}.
\end{align}
Since $\sup_{n\in\integers}\sum_{|k|\leqslant n} \| A_{k} \|_{\infty}$ is finite (because of assumed absolute convergence), the above upper bound tends to 0 as $n,m\to\infty$, since both $\sum_{n\in\integers}A_{n} \odot e_n$ and $\sum_{n\in\integers}B_{n} \odot e_n$ were assumed to converge uniformly.
\end{proof}

\subsubsection{Explicit expressions for generalized Lindbladian and generated dynamics}

\begin{proposition}\label{prop:PFourierConvergence}
Series $\sum_{n\in\integers} P_n \otimes F_n$ converges to $\tilde{P}$ in norm.
\end{proposition}

\begin{proof}
Norm convergence is assured by Proposition \ref{prop:NormConvergenceJA} since $P$ and $P^\prime$ are continuous functions (Proposition \ref{prop:PfunctionConvergence}).
\end{proof}

\begin{proposition}\label{prop:LequivalentForm}
The following hold:
\begin{enumerate}
	\item\label{point:LequivalentFormOne} $\tilde{L} = \sum_{n\in\integers} L_n\otimes F_n$ converging pointwise everywhere in $\iota^{-1}(\mathscr{C}^{0}(\mathbb{T},\matrd_1))$,
	\item\label{point:LequivalentFormTwo} $\tilde{\mathcal{L}}$ admits an equivalent expression
	\begin{equation}\label{eq:LPPadj}
		\tilde{\mathcal{L}} = \tilde{P} (\bar{L} \otimes I - i\Omega \, I \otimes F_z)\tilde{P}^{\dual},
	\end{equation}
	\item\label{point:LequivalentFormThree} $\tilde{L}$ is of standard form.
\end{enumerate}

\end{proposition}

Proof of this Proposition will involve a series of secondary lemmas, accessible in Section \ref{appsection:FourierExpressions} in the Appendix.

\begin{proof}
Ad (\ref{point:LequivalentFormOne}). As $L$ is bounded and therefore square integrable (Lemma \ref{prop:LtSquareInt}), pointwise convergence of the series is assured by Proposition \ref{prop:SreprAconvergence}.

Ad (\ref{point:LequivalentFormTwo}). Let $Q$ denote the Fourier lifting of $\dot{P}$; see Lemma \ref{lemma:ConvergenceDerivativePtPtAdj} for details. Employing Lemmas \ref{lemma:SeriesDerivativePadj} to \ref{lemma:PQadjAntiselfadjoint} we have
\begin{align}\label{eq:LtbyPt}
	\tilde{\mathcal{L}} &= \tilde{P}(\bar{L} \otimes I)\tilde{P}^{\dual} - \tilde{P}(i\Omega \, I \otimes F_z)\tilde{P}^{\dual} \\
	&= \sum_{n,m\in\integers} P_{n} \bar{L} P_{m-n}^{\dual} \otimes F_{m} - \tilde{P}\Big(\tilde{P}^{\dual}(i\Omega \, I \otimes F_z) + Q^\dual \tilde{P}\Big)\tilde{P}^\dual \nonumber \\
	&= \sum_{n,m\in\integers} P_{n} \bar{L} P_{m-n}^{\dual} \otimes F_{m} - i\Omega \, I \otimes F_z - \tilde{P}Q^{\dual}. \nonumber
\end{align}
From \eqref{eq:MMELambda} one has $L_t = \dot{\Lambda}_t \Lambda_{t}^{-1}$, which, together with \eqref{eq:LambdaGeneral}, yield
\begin{equation}
	L_t = \frac{d}{dt}(P_t e^{t\bar{L}}) e^{-t\bar{L}} P_{t}^{-1} = \dot{P}_{t} P_{t}^{\dual} + P_t \bar{L} P_{t}^{\dual}.
\end{equation}
The $m$-th Fourier component of $L_t$, after employing Proposition \ref{prop:PfunctionConvergence}, turns out to be given by series
\begin{equation}\label{eq:LmthFourier}
	L_{m} = \sum_{n\in\integers} P_{n} (\bar{L} + in\Omega) P_{m-n}^{\dual}
\end{equation}
converging in norm due to Lemma \ref{lemma:PXPadjConvergence}. Substituting \eqref{eq:LmthFourier} to \eqref{eq:LtbyPt} we obtain, due to Lemma \ref{lemma:PQadjAntiselfadjoint},
\begin{align}
	\tilde{\mathcal{L}} &= \sum_{m\in\integers} \Big( L_{m} - i\Omega\sum_{n\in\integers} n P_{n} P_{m-n}^{\dual}\Big)\otimes F_m - i\Omega \, I \otimes F_z - \tilde{P}Q^{\dual} \\
	&= \sum_{n\in\integers} L_{n} \otimes F_n - i\Omega \, I \otimes F_z - (Q\tilde{P}^{\dual}+\tilde{P}Q^{\dual}) \nonumber \\
	&= \sum_{n\in\integers} L_{n} \otimes F_n - i\Omega \, I \otimes F_z .\nonumber
\end{align}
Ad (\ref{point:LequivalentFormThree}). Computations are quite straightforward, however lengthy; therefore we will only sketch this part of a proof. Let $Y \in \matrd$ and denote $\delta_{Y} = \ad{Y}$ and $\epsilon_{Y} = \acomm{Y}{\cdot\,}$. By simple algebra, $\delta_Y$ and $\epsilon_Y$ can be lifted to bounded maps over $\matrd_1 \bar{\otimes} L^2(\mathbb{T})$ in such a way that
\begin{equation}
	\delta_Y \otimes I = \ad{Y \otimes e_0}, \qquad \epsilon_{Y} \otimes I = \acomm{Y\otimes e_0}{\cdot\,}.
\end{equation}
Therefore, as $\bar{L} = -i \ad{\bar{H}} + K$ it is easy to see that
\begin{equation}
	(\ad{\bar{H}} \otimes I) (\tilde{f}) = \sum_{n\in\integers} \comm{\bar{H}}{f_n} \otimes e_n = \comm{\bar{H} \otimes e_0}{\tilde{f}}
\end{equation}
leading to
\begin{equation}\label{eq:LequivalentFormPA}
	-i\tilde{P} (\ad{\bar{H}} \otimes I) \tilde{P}^{\dual} = -i \ad{\tilde{p} \, \bar{H} \otimes e_0 \, \tilde{p}^{\adj}}
\end{equation}
which comes from unitarity of $\tilde{p}$. Likewise, employing \eqref{eq:LindInterPicture} we obtain
\begin{equation}\label{eq:LequivalentFormPL}
	\tilde{P} (K\otimes I)\tilde{P}^{\dual} (\tilde{f}) = \sum_{k}\sum_{\{\omega\}}\sum_{k\in\integers} \Big( \Gamma_{k\omega q} \, \tilde{f} \, \Gamma_{k\omega q}^{\adj} - \frac{1}{2}\acomm{\Gamma_{k\omega q}^{\adj}\Gamma_{k\omega q}}{\tilde{f}} \Big),
\end{equation}
where $\Gamma_{k\omega q} = \tilde{p} \, S_{k\omega q} \otimes e_0 \, \tilde{p}^{\dual}$, as can be easily checked by direct computation. Utilizing the chain rule property of $i\Omega \, I \otimes F_z$ and unitarity of $\tilde{p}$ we also have
\begin{align}\label{eq:LequivalentFormPJpPAdj}
	i\Omega \, &\tilde{P}(I\otimes F_z)\tilde{P}^{\adj}(\tilde{f}) = i \Omega \, \tilde{p} [(I\otimes F_z)(\tilde{p}^{\adj} \, \tilde{f} \, \tilde{p})] \tilde{p}^{\adj} \\
	&= i\Omega \, \tilde{p} \, I\otimes F_z(\tilde{p}^{\adj}) \, \tilde{p} \, \tilde{f} + i\Omega \, \tilde{f} \, I\otimes F_z(\tilde{p}) \, \tilde{p}^{\adj} + i\Omega \, I \otimes F_z (\tilde{f}). \nonumber
\end{align}
As $\tilde{p} \, \tilde{p}^{\adj} = I \otimes e_0$ is constant, we have $i\Omega \, I\otimes F_z (\tilde{p} \, \tilde{p}^{\adj}) = 0$ and by the chain rule,
\begin{equation}
	i\Omega \, \tilde{p} \, I\otimes F_z(\tilde{p}^{\adj}) = - i\Omega \, I \otimes F_z (\tilde{p}) \, \tilde{p}^{\adj}.
\end{equation}
We then put $i\Omega \, I \otimes F_z (\tilde{p})$, with aid of \eqref{eq:UnitaryEvolutionOperator} and Schroedinger equation, into a different form
\begin{equation}
	i\Omega \, I \otimes F_z (\tilde{p}) = -i \Big(\sum_{n\in\integers} H_n \otimes e_n \Big) \tilde{p} + i \, \tilde{p} \, \bar{H} \otimes e_0
\end{equation}
which, after substituting back to \eqref{eq:LequivalentFormPJpPAdj} yields
\begin{align}\label{eq:LequivalentFormPIFz}
	i\Omega \, \tilde{P}(I\otimes &F_z)\tilde{P}^{\dual} = i\Omega \, I\otimes F_z - \comm{i\Omega \, I\otimes F_z (\tilde{p}) \, \tilde{p}^{\adj}}{\cdot\,} \\
	&= i\Omega \, I\otimes F_z + i\ad{\sum_{n\in\integers} H_n\otimes e_n} - i\ad{\tilde{p}\,\bar{H}\otimes e_0 \, \tilde{p}^{\adj}}. \nonumber
\end{align}
Finally, equaling \eqref{eq:Lmap} and \eqref{eq:LPPadj} and utilizing \eqref{eq:LequivalentFormPA}, \eqref{eq:LequivalentFormPL} and \eqref{eq:LequivalentFormPIFz} one obtains, after some algebra,
\begin{align}
	\sum_{n\in\integers} L_{n} &\otimes F_n = \tilde{P} (\bar{L} \otimes I) \tilde{P}^{\dual} - i\Omega \, \tilde{P}(I \otimes F_z) \tilde{P}^{\dual} + i\Omega \, I\otimes F_z \\
	&= -i \ad{\sum_{n\in\integers} H_{n}\otimes e_n} + \tilde{P} (\bar{L} \otimes I) \tilde{P}^{\dual} \nonumber
\end{align}
which is of standard form by \eqref{eq:LequivalentFormPL}.
\end{proof}

\section{Summary}

In this paper we presented a formal construction of generalized space of states suited for representing a CP-divisible dynamics of finite-dimensional open quantum systems governed by periodic Lindbladian in Weak Coupling Limit regime. As was shown in previous section, the general solution of MME may be expressed via one-parameter CP-divisible contraction semigroup acting on this space, generated by a generalized, unbounded Lindbladian. We already stressed that this approach shares many similarities with, and therefore is an extension of, unitary (Hamiltonian) time-independent formalism by Howland and others. In the unitary case, one relies on generalized, self-adjoint, infinite dimensional Floquet Hamiltonian, which generates a \emph{group} of unitary evolution operators on the generalized space of states, which by natural construction is also a Hilbert space. The resulting dynamics, after ``projecting'' back on the Hilbert space of a system, remains unitary (see the references for details). The case of CP-divisible dynamics seems to be no different: Floquet Lindbladian $\mathcal{L}$ becomes an analogue of Floquet Hamiltonian and the semigroup which it generates is a contraction $C_0$-semigroup of completely positive and trace preserving maps on Bochner space, while the ``projected'' dynamical map on $\matrd_1$ remains CP-divisible, completely positive and trace preserving, as expected. Shifted Floquet quasienergies, i.e.~the point spectrum of Floquet Hamiltonian, are replaced by point spectrum of $\mathcal{L}$, however they still play similar role in the formalism.

We believe, that applicability of presented theory will parallel and hopefully exceed the applicability of the unitary time-independent approach. Construction proposed in the article seems to be a natural generalization of Howland apparatus to the case of dissipative (irreversible) dynamics, as it takes a full Lindblad-like structure of Markovian Master Equation into account. We emphasize here, that the main benefits of time-independent formalism are still present in our approach, since by introducing the generalized space of states, one similarly replaces a time-dependent problem by an algebraic, time-independent one. As the dynamics on generalized space is given by a \emph{semigroup}, it is in general more straightforward to handle both analytically and numerically, even despite the dimension of space becomes infinite. In such case one can utilize known computational methods to obtain, at least approximate, dynamics in generalized space and, after ``projecting'' back onto matrix space, also dynamics of density matrix itself. These could include various approaches, like e.g.~the \emph{effective Lindbladian theory} (being a generalization of \emph{effective Hamiltonian theory}) and related \emph{van Vleck block diagonalization} of Floquet Lindbladian (see \cite{Ernst2005,Leskes2010,Scholz2010,Ho1983,Chu2004} and references therein for examples in unitary case). These methods already proved to be useful in solving NMR-related problems. The other area of possible applications and probably the most interesting one, includes generalizations onto Markovian Master Equations defined by \emph{quasiperiodic Lindbladians}, i.e.~with many non-commensurate frequencies, where traditional Floquet theory fails. This problem seems to be of non-deniable importance from experimental point of view. We remark here that such generalization was already shown to be possible in case of unitary time-independent formalism (by extending the generalized space of states) at least for Lyapunov-Perron reducible systems (see \cite{Verdeny2016}).

\section*{Acknowledgments}
K.S.~acknowledges support by the National Science Centre, Poland (grant No. 2016/23/D/ST1/02043). R.A.~is supported by the ICTQT through the International Research Agendas Programme (IRAP) of the Foundation for Polish Science (FNP), with structural funds from the European Union (EU). Authors are thankful to anonymous Referee for constructive comments during reviewing of this manuscript, and for suggesting reference \cite{Bach2014}.

\appendix

\section{Technical supplement}

\subsection{Fourier series of matrix-valued functions on \texorpdfstring{$\mathbb{T}$}{a circle}}
\label{appsection:FourierSeries}

We will use few different matrix norms in the following lemmas (listed below). Any two matrix norms are equivalent, i.e.~for arbitrary norms $\| \cdot \|$, $\| \cdot \|^{\prime}$ on $\complexes^{r\times r}$ there always exist some constants $\alpha,\beta > 0$ such that $\alpha \| \cdot \|^{\prime} \leqslant \| \cdot \| \leqslant \beta \| \cdot \|^{\prime}$. We will be using \emph{supremum (operator) induced norm} $\| \cdot\|_{\infty}$, \emph{max norm} $\| \cdot \|_{\mathrm{max}}$, $l^1$ \emph{norm} $\|\cdot\|_{l^1}$ and \emph{Frobenius (Hilbert-Schmidt) norm} $\fnorm{\cdot}$; for matrix $A = [a_{jk}]_{j,k=1}^{r}$ they are defined as follows:
\begin{align}
	\| A \|_{\infty} = \sup_{\| \vec{w} \|_{\complexes^d} \leqslant 1} \| A\vec{w} \|_{\complexes^d}, \qquad \| A \|_{\mathrm{max}} = \max_{j,k}|a_{jk}|, \\
	\| A \|_{l^1} = \sum_{j,k=1}^{r} |a_{jk}|, \qquad \fnorm{A} = \Big( \sum_{j,k=1}^{r} |a_{jk}|^{2} \Big)^{1/2},
\end{align}
where $\| \cdot \|_{\complexes^d}$ stands for arbitrary norm in $\complexes^d$.

\begin{lemma}\label{lemma:FrechetDiffEquivMatElementsDiff}
Let $\mathbf{F}: \reals \to \complexes^{r\times r}$ be given as $\mathbf{F}(t) = [f_{jk}(t)]_{j,k=1}^{r}$, where $f_{jk} : \reals \to \complexes$ and $r \geqslant 1$. Then, $\mathbf{F}$ is differentiable in open interval $(t_1, t_2) \subset \reals$, or $\mathbf{F} \in \mathcal{C}^1((t_1,t_2),\complexes^{r\times r})$ if and only if $f_{jk} \in \mathcal{C}^1 ((t_1, t_2))$ for all pairs $(j,k)$.
\end{lemma}

\begin{proof}
It is enough to notice that vector-valued function
\begin{equation}
	t \mapsto \big( f_{11}(t), \, f_{12}(t) , \, \hdots \, , \, f_{rr}(t) \big) \in \complexes^{r^2}
\end{equation}
is differentiable iff all of its components are differentiable \cite{JanR.Magnus1999}. The result then follows from isometry $\complexes^{r^2} \simeq \complexes^{r\times r}$ and after utilizing equivalence of matrix norms over $\complexes^{r\times r}$.
\end{proof}

\begin{lemma}\label{lemma:SquareIntEquivSqIntMatElem}
Let $\mathbf{F} : U \to \complexes^{r\times r}$, $U \subset \reals$, be given as $\mathbf{F}(t) = [f_{jk}(t)]_{j,k = 1}^r$, where $f_{jk}\in L^1 (U,\nu)$. Then, $\mathbf{F} \in \mathscr{L}^\infty (U,\complexes^{r\times r})$ if and only if $f_{jk}\in L^\infty (U,\nu)$ for all pairs $(j,k)$.
\end{lemma}

\begin{proof}
First, assume $f_{jk} \in L^\infty (U,\nu)$, i.e.~$|f_{jk}(t)|\leqslant \| f_{jk} \|_{L^\infty} < \infty$ for a.e. $t\in U$. From this and from equivalence $\alpha\| \cdot \|_\infty \leqslant \| \cdot \|_{l^1}$ we have
\begin{equation}
	\alpha\| \mathbf{F}(t) \|_{\infty} \leqslant \| \mathbf{F}(t) \|_{l^1} = \sum_{j,k=1}^{r} |f_{jk}(t) | \stackrel{\mathrm{a.e.}}{\leqslant} r^2 \max_{j,k}{\|f_{jk}\|_{L^\infty}} = C
\end{equation}
and so $\| \mathbf{F}(t) \|_{\infty} \leqslant \alpha^{-1}C$ for a.e. $t\in U$, i.e.~$\mathbf{F} \in \mathscr{L}^{\infty}(U,\complexes^{r\times r})$. On the contrary, assume otherwise, i.e.~$\| \mathbf{F}(t) \|_{\infty} \leqslant \| \mathbf{F} \|_{\mathscr{L}^\infty}$ for a.e. $t\in U$; this and equivalence $\beta\| \cdot \|_{\mathrm{max}} \leqslant \| \cdot \|_\infty$ yield
\begin{equation}
	\beta\|\mathbf{F}(t)\|_{\mathrm{max}} \leqslant \| \mathbf{F}(t) \|_{\infty} \stackrel{\mathrm{a.e.}}{\leqslant} \| \mathbf{F} \|_{\mathscr{L}^\infty},
\end{equation}
i.e.~$|f_{jk}(t)| < \beta^{-1}\| \mathbf{F} \|_{\mathscr{L}^\infty}$ for a.e. $t\in U$, or $f_{jk}\in L^\infty (U,\nu)$ for all $(j,k)$.
\end{proof}

\begin{lemma}\label{lemma:UniformConvFourierMatrix}
Let $\mathbf{F}: \mathbb{T} \to \complexes^{r\times r}$, $\mathbf{F}(t) = [f_{jk}(t)]_{j,k=1}^{r}$ be periodic. If all functions $f_{jk}$ admit a uniformly convergent Fourier series, so does $\mathbf{F}$. 
\end{lemma}

\begin{proof}
Let $f_{jk} = \sum_{n\in\integers} f_{jk,n} \odot e_n$ converging uniformly for all pairs $(j,k)$. Define matrix $G_{n} = [f_{jk,n}]_{j,k=1}^{r}$, $n\in\integers$. As $\| \cdot \|_{\infty} \leqslant \fnorm{\cdot}$ and square root is a continuous and strictly increasing function, we easily have
\begin{align}\label{eq:UniformConvergenceFourierMr}
	\lim_{n\to\infty} \sup_{t\in\mathbb{T}} \Big\| \sum_{|k|\leqslant n} &G_k e^{ik\Omega t} - \mathbf{F}(t) \Big\|_{\infty} 	\\
	&\leqslant \Big( \lim_{n\to\infty}\sum_{j,k=1}^{r} \sup_{t\in\mathbb{T}} |D_n * f_{jk}(t) - f_{jk}(t)|^{2} \Big)^{1/2}. \nonumber
\end{align}
As $\sup_{t\in\mathbb{T}} |D_n * f_{jk} (t) - f_{jk}(t)| \to 0$, i.e.~Fourier series of all $f_{jk}$ converge uniformly, the limit in \eqref{eq:UniformConvergenceFourierMr} is 0 and indeed $\sum_{n\in\integers} G_n \odot e_n$ converges to $\mathbf{F}$ uniformly. As we have
\begin{equation}
	G_n = \frac{1}{T}\int_\mathbb{T} \mathbf{F}(t) e^{-in\Omega t} \, dt,
\end{equation}
it is the Fourier series of $\mathbf{F}$ and the claim is shown.
\end{proof}

\begin{lemma}\label{lemma:AbsoluteConvFourierMatrix}
Let $\mathbf{F} : \mathbb{T} \to \complexes^{r\times r}$, $\mathbf{F}(t) = [f_{jk}(t)]_{j,k=1}^{r}$. If all functions $f_{jk}$ admit absolutely convergent Fourier series, so does $\mathbf{F}$.
\end{lemma}

\begin{proof}
Again, let $F_n = [f_{jk,n}]$; we have to show $\sum_{n\in\integers} \| F_n \|_{\infty}$ converges. Assume the absolute convergence for all functions $f_{jk}$, i.e.~$\sum_{n\in\integers} | f_{jk,n}| = C_{jk}$. From equivalence $\| \cdot \|_{\infty} \leqslant \alpha \| \cdot \|_{l^1}$,
\begin{equation}
	\lim_{n\to\infty} \sum_{|l|\leqslant n} \| F_l \|_{\infty} \leqslant \alpha \lim_{n\to\infty} \sum_{|l|\leqslant n} \sum_{j,k=1}^{r} |f_{jk,l}| = \alpha \sum_{j,k=1}^{r} C_{jk},
\end{equation}
so $\sum_{n\in\integers} F_n \odot e_n$ converges absolutely.
\end{proof}

\begin{lemma}\label{lemma:MatrixFunFourierUniConv}
If $\mathbf{A} : \mathbb{T}\to B(\matrd_1)$ is everywhere differentiable and $\dot{\mathbf{A}}$ is piecewise continuous, then $\sum_{n\in\integers}A_n \odot e_n$ converges uniformly to $\mathbf{A}$.
\end{lemma}

\begin{proof}
Notice $B(\matrd_1) \simeq \complexes^{d^2\times d^2}$, so $\mathbf{A}(t)  = [A_{jk}(t)]$, $j,k\in \{ 1 \, ... \, d^2 \}$, where all $A_{jk} : \mathbb{T} \to \complexes$ are periodic and continuous. Lemma \ref{lemma:FrechetDiffEquivMatElementsDiff} yields that if $\mathbf{A}$ is differentiable and of piecewise continuous derivative, then $A_{jk}\in \mathcal{C}^1 (\mathbb{T})$ and $\dot{A}_{jk}$ is piecewise continuous, and therefore absolutely integrable, for all pairs $(j,k)$. As Fourier series of such function is uniformly convergent, this implies $\sum_{n\in\integers} A_n \odot e_n$ also converges uniformly due to Lemma \ref{lemma:UniformConvFourierMatrix}.
\end{proof}

\begin{lemma}\label{lemma:MatrixFunFourierAbsConv}
If $\mathbf{A} : \mathbb{T} \to B(\matrd_1)$ is continuous and $\dot{\mathbf{A}}$ is bounded, then $\sum_{n\in\integers} A_n \odot e_n$ converges absolutely and uniformly to $\mathbf{A}$.
\end{lemma}

\begin{proof}
Again, $B(\matrd_1) \simeq \complexes^{d^2 \times d^2}$, so $\mathbf{A}(t) = [A_{jk}(t)]_{j,k=1}^{d^2}$, where all $A_{jk}$ are periodic and continuous. By Lemma \ref{lemma:SquareIntEquivSqIntMatElem}, all $\dot{A}_{jk}$ are bounded and therefore square integrable; this implies that Fourier series of each $A_{jk}$ converges absolutely and uniformly \cite{Tolstov1976}; by Lemmas \ref{lemma:UniformConvFourierMatrix} and \ref{lemma:AbsoluteConvFourierMatrix}, $\sum_{n\in\integers}A_n \odot e_n$ converges absolutely and uniformly to $\mathbf{A}$.
\end{proof}

\subsection{Convergence of Fourier liftings in \texorpdfstring{$B(\mathscr{L}^2(\mathbb{T},\matrd_1))$}{B(L2)}}
\label{appsection:FourierExpressions}

\begin{lemma}\label{lemma:SeriesDerivativePadj}
We have, that 
\begin{enumerate}
	\item\label{point:SeriesDerivativePadjOne} $i\Omega \sum_{n\in\integers} P_{n}^{\dual} \otimes F_z F_{n}$ converges pointwise to $(i\Omega \, I\otimes F_z) \tilde{P}^{\dual}$,
	\item\label{point:SeriesDerivativePadjTwo} $i\Omega \sum_{n\in\integers} P_{n}^{\dual} \otimes F_{n}F_{z}$ converges pointwise to $\tilde{P}^{\dual} (i\Omega \, I\otimes F_z)$,
\end{enumerate}
everywhere in $\iota^{-1} (\mathscr{C}^{1}(\mathbb{T},\matrd_1))$.
\end{lemma}

\begin{proof}
Ad (\ref{point:SeriesDerivativePadjOne}). Both $(i\Omega \, I\otimes F_z) \tilde{P}^{\dual}$ and $i\Omega \sum_{|k|\leqslant n} P_{k}^{\dual} \otimes F_z F_{k}$ may be shown to satisfy
\begin{subequations}
	\begin{align}
		\iota \circ (i\Omega \, I \otimes F_z) \tilde{P}^{\dual}(\tilde{f})(t) &= \partial (\iota\circ \tilde{P}^{\dual}(\tilde{f}))(t) = \frac{d}{dt} P^{\dual}_{t}(f(t)) \\
		&= \dot{P}^{\dual}_{t}(f(t)) + P_{t}^{\dual} (\dot{f}(t)) , \nonumber
	\end{align}
	\begin{align}
		\iota\circ \Big(i\Omega \sum_{|k|\leqslant n} P_{k}^{\dual} \otimes F_z F_{k} (\tilde{f})\Big)(t) &= \iota \circ (i\Omega \, I \otimes F_z) \sum_{|k|\leqslant n} P_{k}^{\dual} \otimes F_{k} (\tilde{f})(t) \\
		&= (D_n * \dot{P}^\dual )(t)(f(t)) + (D_n * P^\dual)(t)(\dot{f}(t)), \nonumber
	\end{align}
\end{subequations}
yielding, after some manipulations, that for any $\tilde{f}\in\iota^{-1}(\mathscr{C}^{1}(\mathbb{T},\matrd_1))$
\begin{align}\label{eq:ConvergenceDerivativePadj}
	&\Big\| i\Omega \sum_{|k|\leqslant n} P_{k}^{\dual} \otimes F_z F_{k}(\tilde{f}) - (i\Omega \, I\otimes F_z) \tilde{P}^{\dual}(\tilde{f}) \Big\|_{\mathscr{L}^{2}_{1}}^{2} \\
	&\leqslant \sup_{t\in\mathbb{T}}\tnorm{f(t)}^{2}\cdot \frac{1}{T} \int_{\mathbb{T}} \left\|(D_n * \dot{P}^\dual)(t) - \dot{P}^{\dual}_{t}\right\|_{\infty}^{2} dt \nonumber \\
	&+ \sup_{t\in\mathbb{T}}\|\dot{f}(t)\|_{1}^{2} \cdot \frac{1}{T}\int_{\mathbb{T}} \left\|(D_n * P^\dual ) (t) - P_{t}^{\dual}\right\|_{\infty}^{2} dt.\nonumber
\end{align}	
as $f$ and $df/dt$ are continuous and periodic and attain their maxima. As $P^\dual$ and $\dot{P}^\dual$ are square integrable, their Fourier series converge in $\mathscr{L}^{2}(\mathbb{T},B(\matrd_1))$ by Theorem \ref{FourierSeriesConvLp} and Proposition \ref{prop:PfunctionConvergence} and therefore this upper bound is 0 as $n \to \infty$; the first claim is shown.

Ad (\ref{point:SeriesDerivativePadjTwo}). Analogously, we start with noting that
\begin{subequations}
	\begin{equation}
		\iota \circ \tilde{P}^{\dual} (i\Omega \, I \otimes F_z)(\tilde{f})(t) = P_{t}^{\dual}(\dot{f}(t)),
	\end{equation}
	\begin{align}
		\iota \circ \Big( \sum_{|k|\leqslant n} P_{k}^{\dual} \otimes F_{k} F_z \Big)(\tilde{f})(t) &= \iota\circ \sum_{|k|\leqslant n} P_{k}^{\dual} \otimes F_{k} (i\Omega \, I \otimes F_z)(\tilde{f})(t) \\
		&= (D_n * P^\dual) (t) (\dot{f}(t)).\nonumber
	\end{align}
\end{subequations}
As $df/dt$ is continuous and periodic, it is bounded and so one gets, after some algebra,
\begin{align}
	\Big\| i\Omega \sum_{|k|\leqslant n} P_{k}^{\dual} &\otimes F_z F_{k}(\tilde{f}) - \tilde{P}^{\dual} (i\Omega \, I\otimes F_z)(\tilde{f}) \Big\|_{\mathscr{L}^{2}_{1}}^{2} \\ 
	&\leqslant \sup_{t\in\mathbb{T}}\|\dot{f}(t)\|_{1}^{2} \cdot \frac{1}{T}\int_{\mathbb{T}} \left\| (D_n * P^\dual) (t) - P_{t}^{\dual} \right\|^{2}_{\infty} dt \nonumber
\end{align}
which again vanishes as $n\to\infty$, since $P^\adj$ is square integrable and the second claim is proved.
\end{proof}

\begin{lemma}\label{lemma:ConvergenceDerivativePtPtAdj}
Denote the Fourier lifting of $\dot{P}$ by $Q$. Then, we have
\begin{equation}
	Q = i\Omega \sum_{n\in\integers} n P_{n} \otimes F_{n}, \quad Q^{\dual} = i\Omega \sum_{n\in\integers} n P_{n}^{\dual} \otimes F_{n}
\end{equation}
converging in norm.
\end{lemma}

\begin{proof}
As Fourier series of both $\dot{P}$ and $\dot{P}^\dual$ converge uniformly by Proposition \ref{prop:PfunctionConvergence}, we can apply Proposition \ref{prop:NormConvergenceJA}.
\end{proof}

\begin{lemma}\label{lemma:JdeltaPadjCommRel}
It holds, that $\comm{i\Omega \, I\otimes F_z}{\tilde{P}^{\dual}} = Q^{\dual}$.
\end{lemma}

\begin{proof}
From proposition \ref{prop:FourierOpComm} we have $F_z F_{k} = k F_{k} + F_{k} F_{z}$ for any $k\in\integers$ and we can write
\begin{equation}\label{eq:SeriesDerivativePadj}
	i\Omega \sum_{|k|\leqslant n} P_{k}^{\dual} \otimes F_z F_k = i \Omega \sum_{|k|\leqslant n} k P_{k}^{\dual} \otimes F_{k} + i\Omega \sum_{|k|\leqslant n} P_{k}^{\dual} \otimes F_k F_z,
\end{equation}
where all the series converge due to Lemmas \ref{lemma:SeriesDerivativePadj} and \ref{lemma:ConvergenceDerivativePtPtAdj}. Therefore, putting $n\to\infty$ we can restate \eqref{eq:SeriesDerivativePadj} as
\begin{equation}
	(i\Omega \, I \otimes F_z ) \tilde{P}^{\dual} = Q^{\dual} + \tilde{P}^{\dual} (i\Omega \, I \otimes F_z ),
\end{equation}
yielding the claim.
\end{proof}

\begin{lemma}\label{lemma:PXPadjConvergence}
We have
\begin{equation}\label{eq:PXPadjConvergence}
	\tilde{P} (\bar{L}\otimes I)\tilde{P}^{\dual} = \sum_{n,m\in\integers} P_{n} \bar{L} P_{m-n}^{\dual} \otimes F_{m}
\end{equation}
converging in norm.
\end{lemma}

\begin{proof}
Consider \emph{constant} function $t\mapsto A_t$ given via $A_t (a) = \bar{L}(a)$, $a \in \matrd$. Then, $A = \bar{L} \odot e_0$, and $\tilde{A} = \bar{L}\otimes F_0$, which is bounded. Let $\beta$ be defined on $\mathscr{L}^{2}(\mathbb{T},\matrd_1)$ by
\begin{equation}
	\beta(f)(t) = \bar{L} P_{t}^{\dual} (f(t)).
\end{equation}
Then, $\tilde{\beta} = (\bar{L}\otimes I)\tilde{P}^{\dual}$. As $P$ admits absolutely and uniformly convergent Fourier series (Proposition \ref{prop:PfunctionConvergence}), so does the function $t\mapsto \bar{L}P_{t}^{\dual}$. Using group properties of operators $F_n$ and applying Proposition \ref{prop:NormConvergenceProductofFunctions}, we have
\begin{equation}
	\tilde{P}\tilde{\beta} = \sum_{n,m\in\integers} P_{n} \bar{L} P_{m}^\dual \otimes F_{n+m}
\end{equation}
converging in norm; \eqref{eq:PXPadjConvergence} appears after changing $n+m \to m$.
\end{proof}

\begin{lemma}\label{lemma:MinusPPartialPadjConv}
The following equalities hold (series converge in norm):
\begin{equation}\label{eq:MinusPQadj}
	i\Omega \sum_{m,n\in\integers} m P_{n} P_{m}^{\dual} \otimes F_{n+m} = \tilde{P}Q^\dual , \quad i\Omega \sum_{m,n\in\integers} n P_{n} P_{m}^{\dual} \otimes F_{n+m} = Q\tilde{P}^\dual .
\end{equation}
\end{lemma}

\begin{proof}
This is immediate from Moore-Smith theorem. Put $(s_{mn})$ as
\begin{equation}
	s_{mn} = i\Omega \sum_{|k|\leqslant m} \sum_{|l|\leqslant n} k P_{l} P_{k}^{\dual} \otimes F_l F_k .
\end{equation}
Then, one easily shows $y_m = i\Omega \tilde{P} \sum_{|k|\leqslant m}k P_{k}^{\dual} \otimes F_k$ satisfies
\begin{equation}
	\| y_m - s_{mn} \|_\infty \leqslant \Omega \Big\| \tilde{P} - \sum_{|l|\leqslant n} P_{l} \otimes F_l\Big\|_{\infty} \Big\|\sum_{|k|\leqslant m} k P_{k}^{\dual} \otimes F_k \Big\|_{\infty},
\end{equation}
and so $\lim_{n\to\infty} s_{mn} = y_m$ as $\sum_{n\in\integers} P_{n} \otimes F_n$ converges to $\tilde{P}$. On the other hand, $z_n = i\Omega \sum_{|l|\leqslant n} P_{l} \otimes F_l \, Q^\dual$ satisfies
\begin{equation}
	\| z_n - s_{mn} \|_\infty \leqslant \Omega \Big\| \sum_{|l|\leqslant n} P_{l} \otimes F_l \Big\|_{\infty} \Big\| Q^{\dual} - i\Omega\sum_{|k|\leqslant m} k P_{k}^{\dual} \otimes F_k\Big\|_{\infty},
\end{equation}
so $\lim_{n\to\infty} s_{mn} = z_n$ since $i\Omega\sum_{n\in\integers} n P_{n}^{\dual} \otimes F_n$ converges to $Q^{\dual}$ (by Lemma \ref{lemma:ConvergenceDerivativePtPtAdj}). Limits of both $(y_m)$ and $(z_n)$ coincide and are equal to $\tilde{P} Q^{\dual}$ and Moore-Smith theorem yields the convergence of $(s_{mn})$, as claimed. Second claim follows from taking dual of first equality and renaming indices.
\end{proof}

\begin{lemma}\label{lemma:PQadjAntiselfadjoint}
Operator $\tilde{P}Q^{\dual}$ satisfies $\tilde{P}Q^{\dual} + Q\tilde{P}^{\dual}=0$.
\end{lemma}

\begin{proof}
For any $\tilde{f} \in \matrd_1 \bar{\otimes} L^2(\mathbb{T})$ and $t \in \mathbb{T}$ we obtain after simple algebra,
\begin{align}
	\iota \circ (Q\tilde{P}^{\dual} + \tilde{P}Q^{\dual})(\tilde{f})(t) &= (\dot{P}_{t} P_{t}^{\dual} + P_t \dot{P}^{\dual}_{t})(f(t)) \\
	&= \Big(\frac{d}{dt}P_t P_{t}^{\dual}\Big)(f(t)) = 0, \nonumber
\end{align}
after applying unitarity of $P$ (point \ref{point:PpropertiesUnitary} in Proposition \ref{prop:PfunctionConvergence}).
\end{proof}

\end{document}